%% file: dist-plan.tex
\def\podc{0} 
\def\cready{0} 
\def\full{1} 
\def\jformat{0} 
\ifnum\cready=0
\ifnum\jformat=0
\documentclass[11pt]{article}
\usepackage{fullpage}
\else
\documentclass[twocolumn]{svjour3}
\smartqed
\fi
\usepackage[colorlinks,urlcolor=blue,citecolor=blue,linkcolor=blue]{hyperref}
\else
\documentclass[sigconf]{acmart}

\copyrightyear{2018}
\acmYear{2018}
\setcopyright{acmcopyright}
\acmConference[PODC '18]{ACM Symposium on Principles of Distributed
Computing}{July 23--27, 2018}{Egham, United Kingdom}
\acmBooktitle{PODC '18: ACM Symposium on Principles of Distributed
Computing, July 23--27, 2018, Egham, United Kingdom}
\acmPrice{15.00}
\acmDOI{10.1145/3212734.3212748}
\acmISBN{978-1-4503-5795-1/18/07}

\fi

\usepackage{amsfonts,amsmath,mathtools}
\usepackage{comment}
\usepackage{paralist}

\usepackage{algorithmic,algorithm}
\usepackage{bbm,bm}
\usepackage{xspace,prettyref}

\usepackage{color}

\usepackage{dsfont}

\usepackage{graphicx}
\usepackage{epstopdf}
\graphicspath{{./figs/}}
\usepackage{caption}
\usepackage{subcaption}

\ifnum\jformat=0
\newtheorem{theorem}{Theorem}
\newtheorem{lemma}{Lemma}
\newtheorem{claim}[lemma]{Claim}
\newtheorem{definition}[lemma]{Definition}
\newtheorem{corollary}[lemma]{Corollary}

\newtheorem{remark}{Remark}
\fi
\newtheorem{ourclaim}[lemma]{Claim}

\newcommand{\BPFOF}{\smallskip \begin{proofof}} \newcommand {\EPFOF}{\end{proofof}}

\ifnum\cready=0
\def\congest{{\sf CONGEST}}
\def\ver{{\sf v}}
\def\calP{{\cal P}}
\def\calG{{\cal G}}
\def\calF{{\cal F}}
\def\calT{{\cal T}}
\def\calu{{\sf u}}
\def\calv{{\sf v}}
\def\caly{{\sf y}}
\def\Ex{{\rm Ex}}
\else
\def\congest{{\sf CONGEST}}
\def\ver{{\sf v}}
\def\calP{{\mathcal P}}
\def\calG{{\mathcal G}}
\def\calF{{\mathcal F}}
\def\calT{{\mathcal T}}
\def\calu{{\mathsf u}}
\def\calv{{\mathsf v}}
\def\caly{{\mathsf y}}
\def\Ex{{\mathrm Ex}}
\fi

\ifnum\cready=0
\ifnum\jformat=0
\newcommand{\qed}{\;\;\;\FullBox}
\newenvironment{proof}{\noindent{\bf Proof:~~}}{\(\qed\)}
\fi
\fi

\def\eps{\epsilon}

\def\poly{{\rm poly}}

\def\FullBox{\hbox{\vrule width 8pt height 8pt depth 0pt}}

\newenvironment{proofof}[1]{\smallskip\noindent{\bf Proof of #1:}}%
        {\hspace*{\fill}$\Box$\par}

\newcommand{\mnote}[1]{{\color{red}($\star$)}\marginpar{\tiny\bf
		\begin{minipage}[t]{0.5in}
			\raggedright#1
		\end{minipage}}}

\definecolor{darkgreen}{rgb}{0.0, 0.5, 0.0}

\newcommand{\Dana}[1]{{\color{darkgreen}\bf Dana: #1}}

\ifnum\jformat=1
\journalname{Distributed Computing}
\fi
\begin{document}

\ifnum\podc=1 
\begin{titlepage}
\title{Property Testing of Planarity in the \congest\ model \\ \medskip {\large [Regular paper]}}
	\author{
		Reut Levi\thanks{Weizmann Institute of Science, {\tt reut.levi@weizmann.ac.il}.}
				\and  Moti Medina\thanks{Department of Electrical \& Computer Engineering, Ben-Gurion University of the Negev, {\tt medinamo@bgu.ac.il}.}
				\and  Dana Ron
		\thanks{School of Electrical Engineering, Tel Aviv University, {\tt danaron@tau.ac.il}.
			This research was partially supported
			by the Israel Science Foundation grant No. 671/13.}
	}
\date{}
\else%
\ifnum\cready=1
\title{Property Testing of Planarity in the \congest\ model}
\subtitle{Regular paper}
\author{Reut Levi}
\affiliation{%
  \institution{
    Weizmann Institute of Science}
}
\email{reut.levi@weizmann.ac.il}
\author{Moti Medina}
\affiliation{%
  \institution{
    Ben Gurion University of the Negev}
}
\email{medinamo@bgu.ac.il}
\author{Dana Ron}
\authornote{This research was partially supported 			by the Israel Science Foundation grants No. 671/13 and 1146/18.}
\affiliation{%
  \institution{
    Tel Aviv University}
}
\email{danaron@tau.ac.il}
\else 
\ifnum\jformat=0
\title{\bf{Property Testing of Planarity in the \congest\ model}\thanks{This
article extends work presented at PODC 2018~\cite{LMR18}.}}
	\author{
		Reut Levi\thanks{Weizmann Institute of Science, {\tt reut.levi@weizmann.ac.il}. This project has received
funding from the European Research Council (ERC) under the European Unions Horizon 2020 research and innovation
programme (grant agreement No. 819702). It was also supported by ERC-CoG grant 772839.}
				\and  Moti Medina\thanks{Department of Electrical \& Computer Engineering, Ben-Gurion University of the Negev, {\tt medinamo@bgu.ac.il}.}
				\and  Dana Ron
		\thanks{School of Electrical Engineering, Tel Aviv University, {\tt danaron@tau.ac.il}.
			This research was partially supported
			by the Israel Science Foundation grants No. 671/13 and 1146/18.}}
\date{}
\else 
\title{\bf{Property Testing of Planarity in the \congest\ model}\thanks{This
article extends work presented at PODC 2018~\cite{LMR18}. Reut Levi is partially supported by the European Research Council (ERC) under the European Unions Horizon 2020 research and innovation programme (grant agreement No. 819702). Reut Levi was also supported by ERC-CoG grant 772839. Dana Ron was partially supported by the Israel Science Foundation grants No. 671/13 and 1146/18.}}
\author{Reut Levi         \and
        Moti Medina       \and
        Dana Ron
}

\authorrunning{R.~Levi, M.~Medina and D.~Ron}

\institute{R. Levi \at
              Weizmann Institute of Science \\
              \email{reut.levi@weizmann.ac.il}           
           \and
           M. Medina \at
              School of Electrical \& Computer Engineering, Ben-Gurion University of the Negev \\
              \email{medinamo@bgu.ac.il} \\
              Orcid ID: \url{https://orcid.org/0000-0002-5572-3754
}
           \and
           D. Ron \at
                School of Electrical Engineering, Tel Aviv University \\
                \email{danaron@tau.ac.il}
}
\fi 
\fi 
\fi 

\ifnum\cready=0
\maketitle
\fi

\ifnum\jformat=1
\date{Received: date / Accepted: date}
\maketitle
\fi

\begin{abstract}


We give a distributed  algorithm in the {\sf CONGEST} model for property testing  of planarity with one-sided error in general (unbounded-degree) graphs.
Following Censor-Hillel et al. (DISC 2016), who recently initiated the study of property testing in the distributed setting, our algorithm gives the following guarantee:
For a graph $G = (V,E)$ and a distance parameter $\eps$, if $G$ is planar, then every node outputs {\sf accept\/}, and if $G$ is $\eps$-far from being planar (i.e., more than $\eps\cdot |E|$ edges need to be removed in order to make $G$ planar), then with 
probability $1-1/\poly(n)$
at least one node outputs {\sf reject}.
The algorithm runs in $O(\log|V|\cdot\poly(1/\eps))$ rounds, and we show that this result is tight in terms of the dependence on $|V|$.

Our algorithm combines several techniques of graph partitioning and local verification of planar embeddings.
Furthermore, we show how a main subroutine in our algorithm can be applied to derive additional results for property testing  of cycle-freeness and bipartiteness, as well as the construction of spanners, in minor-free
(unweighted) graphs.
\ifnum\jformat=1
\keywords{Distributed Graph algorithms \and Congest \and Distributed property testing \and Planarity testing}
\fi
\end{abstract}

\ifnum\podc=1
\thispagestyle{empty}
\end{titlepage}
	
\setcounter{page}{1}
	\newpage	
\else
\ifnum\cready=1 

\ccsdesc[500]{Theory of computation~Distributed algorithms}
\ccsdesc[300]{Theory of computation~Distributed computing models}
\ccsdesc[300]{Theory of computation~Graph algorithms analysis}

\keywords{Distributed algorithms; Congest; Distributed property testing; Planarity testing}
\fi 
\fi 

\ifnum\cready=1
\maketitle
\fi

\input{intro.tex}


\input{alg}

\input{lb}

\input{second-partition}

\ifnum\jformat=0
\subsection*{Acknowledgements}
\else
\begin{acknowledgements}
\fi
We would like to thank Mohsen Ghaffari and Merav Parter for helpful information.
We would also like to thank the anonymous PODC reviewers for their helpful comments.
\ifnum\jformat=1
\end{acknowledgements}
\fi

\ifnum\cready=0
\bibliographystyle{plain}
\else
\bibliographystyle{ACM-Reference-Format}
\fi
\bibliography{refs}

\ifnum\podc=1
\appendix

\section{Related work in (centralized) property testing}\label{a:related}
\ptestingrel

\section{Missing details for Section~\ref{sec:plan-test}}\label{a:partition}
\EmulateMerge

\subsection{Missing proof of Claim~\ref{clm:merge}}
\mergestepclaim

\subsection{Missing proof of Claim~\ref{clm:stage-I-correct1}}
\stagecomplete

\subsection{Missing proof of Claim~\ref{clm:stage-I-correct2}}
\connecdiam

\section{Missing details for Section~\ref{sec:minor-free}}\label{a:parti2}
\secondparti

\subsection{Missing proof of Corollary~\ref{cor:cycle-free-bip}}
\coroapp

\lbproof
\fi

\ifnum\full=1
\appendix
\input{app}
\fi

\end{document}

%% file: intro.tex
\section{Introduction}\label{sec:intro}
Planarity is an important and well studied property of graphs.
In the setting of centralized algorithms, there are several algorithms that run in linear time for deciding whether a graph is planar~(e.g.,~\cite{LEC66,HT74,BM04}). 
In the context of distributed  algorithms in the \congest~\cite{P00} (and even {\sf LOCAL}~\cite{L92}) model, 
the number of rounds must be at least linear in the diameter of the graph (for any deterministic or randomized one-sided error algorithm).
One begging question is whether there exists an algorithm (in the \congest\ model) for deciding planarity whose round complexity matches (or is not too far) from this lower bound.\footnote{We note that Ghaffari and Haeupler~\cite{GH16} consider a different, but related question of finding a planar embedding of a planar graph. They give a distributed algorithm for this problem using $O(D\cdot \min\{\log n,D\})$ rounds
(where $D$ is the diameter and $n$ is the number of nodes).}
Another question is whether there exists a natural relaxation of this decision problem, which allows to obtain round complexity that does not depend on the diameter.

In this work we address the latter question, by considering the relaxation of {\em property testing\/} in
the distributed setting.
In all that follows, unless explicitly stated otherwise, when we refer to distributed algorithms, we mean in the \congest\ model.
%
Following Censor-Hillel et al.~\cite{CFSV16},  who recently initiated the study of distributed property testing, we require the following from the algorithm.  Let $G = (V,E)$ be a graph over $n$ nodes and $m$ edges.
If 
$G$
is planar, then all
nodes should output {\sf accept\/}, while if $G$ is $\eps$-far from being planar (i.e., more than $\eps\cdot m$ edges should be removed in order to make $G$ planar), then at least one node should output {\sf reject\/}.
The algorithm is allowed a bounded error probability,
where if it errs only on graphs that are $\eps$-far from being planar, then it is said to have {\em one-sided error\/}.%
\ifnum\podc=0
\footnote{Observe that if the algorithm has one-sided error and a constant error probability, then its error probability can be reduced to $\delta$, for a given parameter $\delta$, at a multiplicative cost of $\log(1/\delta)$ in the number of rounds. We are able to obtain error probability $1/\poly(n)$ without this extra cost.}
\fi

\sloppy
Our main result is a distributed one-sided error property testing algorithm for planarity that runs in $O(\log(n)\poly(1/\eps))$ rounds and succeeds with probability $1- 1/\poly(n)$.
We also show that $\Omega(\log n)$ rounds are necessary for any such algorithm and constant $\eps$ (even if the algorithm is allowed a constant error probability), implying that our result is tight up to the dependence on $\epsilon$.

In the context of (centralized) property testing, there is a line of work 
on two-sided error testing of planarity in bounded-degree graphs~\cite{BSS08,HKNO09,LR15,KSS19}.
The best known algorithm~\cite{KSS19} performs polynomial number of queries in $d$ and $\eps$, where $d$ is the degree bound (and succeeds with high constant probability).
The best one-sided error testing algorithm~\cite{KSS18}, for bounded-degree graphs, has query complexity $n^{1/2+o(1)}$, which is almost optimal.
There is no known sublinear testing algorithm for unbounded-degree graphs.
Note that in contrast,
in the distributed setting we are able to obtain an optimal (in terms of the dependence on $n$)
algorithm that has one-sided error and works for unbounded-degree graphs.

In addition to our main result, we show that under the promise that $G$ is planar (and more generally,
minor-free\footnote{Recall that a graph $H$
is a {\em minor\/} of a graph $G$ if $H$ is isomorphic to a graph that can be obtained
by zero or more edge contractions on a subgraph of $G$.
We say that a graph $G$ is
{\em $H$-minor free\/} (or {\em excludes $H$ as a minor\/}) if $H$ is not a minor of $G$.
 We say that $G$ is ``minor-free'' if it is $H$-minor free for
some fixed $H$ of constant size.}
for any fixed minor),
we can use our techniques to obtain other distributed property testing algorithms as
well as an algorithm for the construction of spanners.\footnote{A spanner of a graph $G$ is a
(sparse) subgraph
of $G$ that maintains distances up to a multiplicative factor, $s$, which is called the {\em stretch\/} factor,
and the spanner is referred to as an $s$-spanner.}

In the next two subsections we discuss our results in more detail.

\subsection{A high-level description of our algorithm for testing planarity}
\label{subsec-intro-high-level}
The algorithm works in two stages.
The goal of the first stage (which is deterministic) is to
partition the nodes of $G$ into parts
for which the following holds: (1) Each part is connected and has diameter $\poly(1/\eps)$;
 (2) The total number of edges between parts is at most $\eps m/2$.
 In the course of this stage, some node(s) may obtain evidence that the graph is not planar, and output {\sf reject}.
 This evidence is in the form of messages  received that are not consistent with the execution of the algorithm on
 any planar graph. 
   Conditioned on this stage completing successfully, if $G$ is $\eps$-far from being planar, then the subgraph induced by at least one of the parts in the partition is $(\eps/2)$-far from being planar. The goal of the second stage is to search for evidence in each part to non-planarity, by exploiting the fact that the diameter 
 of each part is small.

 We next give some more details about each stage of the algorithm.
 Before doing so we recall several notions and basic facts.
 A {\em forest decomposition\/} of a graph is a partition of its edges into forests.
 The {\em arboricity\/} of a graph is the
 minimum number of forests into which its edges can be partitioned.
 Any planar graph has arboricity at most $3$, 
 and if we perform any sequence of contractions of edges on a planar graph, then we obtain a planar graph.

\paragraph{The first stage.}
 \sloppy The first stage consists of $t=O(\log(1/\eps))$ phases. At the start of phase $i$, the nodes are
 partitioned into $k_i$ parts, denoted $P_i^1,\dots,P_i^{k_i}$
 where each part is connected and has diameter 
 at most $4^{i-1}$.  
 Let us denote this partition by
 $\calP_i$.
  In the initial partition, $\calP_1$, each part simply consists of a single node. For each phase, let
  $\calG_i$ denote the auxiliary weighted graph that results from contracting  each part $P_i^j$ into a single node, which we denote by $\ver(P_i^j)$. The weight of an edge
 $(\ver(P_i^j),\ver(P_i^{j'}))$ is the number of edges in $G$ with one endpoint in $P_i^j$ and the other in $P_i^{j'}$.

 Each phase starts by emulating the
 (deterministic distributed) forest decomposition algorithm of Barenboim and Elkin~\cite{BE10} on $\calG_i$ (ignoring the weights).  This algorithm works in $O(\log(n))$ rounds and gives the following guarantee. If $\calG_i$ has arboricity at most $\alpha$, then the algorithm provides a forest decomposition into at most $3\alpha$ forests. On the other hand, if the algorithm fails in defining such a decomposition, then at least one node (in $\calG_i$ and hence in $G$) has evidence that $\calG_i$  has arboricity greater than $\alpha$.


 Following the above {\em forest decomposition step\/}, the algorithm executes a 
 {\em merging step\/} (based on the clustering method of
 Czygrinow, Ha\'{n}\'{c}kowiak, and Wawrzyniak~\cite{CHW08}).
 In this step, parts of $\calP_i$ are merged, thus creating $\calP_{i+1}$.
%
 This merging procedure satisfies the following: (1) The maximum diameter 
 of the parts in $\calP_{i+1}$ is at most a constant factor larger than the maximum diameter 
 of the parts in $\calP_i$; (2)
 The 
 number of edges in $G$ between parts in $\calP_{i+1}$ (the total 
 weight of
  edges in $\calG_{i+1}$) is a constant fraction of the number of edges between parts in $\calP_i$
 (the total weight of edges in $\calG_i$).  
 The latter crucially
 relies on the bounded arboricity of $\calG_i$, which is ensured by the forest decomposition step.

\smallskip
Thus, the following central feature of (the first stage of) our algorithm emerges. Though we do not have a promise that the underlying graph is planar (indeed, it may be far from being planar), we are able to build on algorithms that work under the promise that the graph is planar. Namely, in each phase we verify that $\calG_i$ has bounded arboricity, by running an algorithm that works under the promise that the underlying graph has bounded arboricity. Failure of this algorithm is detected by at least one node in the graph. On the other hand, if this algorithm succeeds, then we are ensured that we shall make the desired progress in the transformation from $\calG_i$ to $\calG_{i+1}$ (in terms of the decrease in the total weight of edges).


\paragraph{The second stage.}
Assume the first stage completed successfully (where this always holds if $G$ is planar), and let $\calP = (P^1,\dots,P^k)$ be the final partition of the nodes (i.e., $\calP = \calP_t$).
Recall that the subgraph induced by each part is connected, and has diameter $\poly(1/\eps)$.
Furthermore, the first stage ensures the following for each part $P^j$: (1) There is a designated
{\em root\/} vertex in $P^j$, denoted $r^j$, where
each node in $P^j$ knows the identity of $r^j$; (2) There is an underlying spanning tree
in $P^j$, rooted at $r^j$, where each node in $P^j$ knows which of its incident edges is also incident to its parent in the tree, and which edges are incident to its children.

The second stage consists of two steps. In the first step, the (deterministic distributed) planar embedding algorithm of Ghaffari and Haeupler~\cite{GH16} is emulated on the subgraph induced by each part $P^j$, denoted
$G^j$.
The planar embedding algorithm works under the promise that $G^j$ is planar, and when it completes, each node in $G^j$ has a circular ordering over its incident edges that corresponds to a planar embedding (what is known as a {\em combinatorial embedding}).
Since the diameter of each part is $\poly(1/\eps)$,
the number of rounds performed by this algorithm is $\poly(1/\eps)$.
If this step fails in determining an ordering for all nodes (in the aforementioned number of rounds), then this constitutes evidence that $G^j$ is not planar.
However, it is possible that an ordering is determined though $G^j$ is not planar.

Hence, the second step in this stage is aimed at detecting non-planarity of some $G^j$ given the ordering provided by the~\cite{GH16} algorithm.\footnote{It was communicated to us by one of the authors of~\cite{GH16} that their algorithm can be modified so as to detect if the underlying graph is not planar~\cite{H-p-com}.
For the sake of a self-contained presentation, we rely on version of the algorithm as provided in~\cite{GH16} (which works under the promise that the graph is planar), and check that the ordering computed by~\cite{GH16} is consistent with a planar embedding.}
More precisely, as noted previously, if $G$ is $\eps$-far from being planar, then 
at least one $G^j$ is $(\eps/2)$-far from being planar. 
Using the ordering of edges incident to each node in $G^j$ together with a BFS tree
rooted at $r^j$,  denoted $T_B^j$ (which can be constructed in $\poly(1/\eps)$ rounds),
we define a certain condition on each of the non-tree edges of $T_B^j$. We show that if
 $G^j$ is far from being planar, then there are relatively many non-tree edges in $G^j$ that violate this condition, while if $G^j$ is planar, then no non-tree edge violates this condition. Furthermore, given a violating edge, it is possible to detect violation in 
 $\poly(1/\eps)$
 rounds.\footnote{We note that a previous version of the definition of this condition, given in the conference version of this paper~\cite{LMR18}, contained an error, which is fixed in the current version.}
 Hence, by sampling
$\Theta(\log(n)/\eps)$
 non-tree edges in each $G^j$ and running the detection procedure on each, a violation is detected with probability
 $1-1/\poly(1/n)$.


We note that, as shown in~\cite{ELM17,DISC17}, the algorithm of Elkin and Neiman~\cite{EN17} can be adapted to
obtain with high probability a partition of the nodes into parts of diameter $O(\log(n)/\eps)$ such that the number of edges between parts is at most $\eps m$. Replacing Stage~I in our algorithm with this procedure (and running Stage~II on each part), results in a testing algorithm that runs in $O(\log^2(n)\poly(1/\eps))$
(while our algorithm runs in $O(\log(n)\poly(1/\eps))$ rounds).

\subsection{Implications and applications for minor-free graphs}\label{subsec:intro-minor-free}
Suppose we have a promise that $G$ is planar, or more generally, minor-free for any constant size minor.
In such a case, the first stage of our algorithm always ensures that the nodes of $G$ are partitioned into
parts with diameter $\poly(1/\eps)$ such that the number of edges between parts is at most $\eps m/2$.
Such a partition can be used for testing properties such as cycle-freeness and bipartiteness (and more generally, hereditary properties that can be tested in a number of rounds that is linear (or even polynomial) in the diameter).
Thus, for cycle-freeness and bipartiteness we obtain a {\em deterministic\/} testing algorithm that runs in
$O(\log(n)\poly(1/\eps))$ rounds.
Such a partition can also be used to obtain $\poly(1/\eps)$-spanners (for unweighted graphs)
(deterministically)
in $O(\log(n)\poly(1/\eps))$ rounds. 

\sloppy
We also show how to modify the partition algorithm so as to obtain a tradeoff between the round-complexity and the success probability of the algorithm. More precisely, with probability at least $1-\delta$, the modified algorithm
gives the same guarantee as above for the partition
 in $O(\log(1/\eps)(\log^*(n)+\log(1/\delta)))$ rounds.
 The complexity of the testing algorithms and spanner construction algorithm are improved accordingly
 (see Corollaries~\ref{cor:cycle-free-bip} and~\ref{cor:spanner}, respectively).
 Finally, if
constant success probability suffices, then it is possibly to remove the dependence on $n$ completely.

The testing results can be compared with the $\Omega(\log n)$ lower bound of Censor-Hillel et al.~\cite{CFSV16}
for distributed testing of these properties on general (bounded-degree) graphs (with constant success probability).
\sloppy
The spanner result can be compared to the recent result of Elkin and Neiman~\cite{EN17}. 
They provide a $k$-round distributed algorithm for general (unweighted) graphs that with probability $1-\delta$ constructs a $(2k-1)$-spanner with $O(n^{1+1/k}/\delta)$ edges. 
In order to obtain an {\em ultra-sparse spanner}, namely, a spanner of size $n(1+o(1))$ (with probability $1-o(1))$, it is necessary to set $k = \omega(\log n)$.
In our context, of minor-free graphs, we can obtain an ultra-sparse spanner deterministically for any
$\eps = o(1)$, which allows to construct ultra-sparse spanners with stretch $s$
 for any $s = \omega(1)$.

\input{related}

\ifnum\podc=1
\subsubsection*{Organization}
In Section~\ref{sec:plan-test} we describe and analyze the algorithm for testing planarity, and in Section~\ref{sec:minor-free} we state the additional results regarding implications and applications for minor-free graphs. All missing details for these sections, as well as the lower bound, can be found in the appendices.
\fi

%% file: related.tex
\def\local{{\sf LOCAL}}

\subsection{Related work}
\ifnum\podc=1
In this section we overview results in distributed property testing, 
as well as distributed algorithms
that operate on a graph with a promised property. We overview related results in centralized property testing in Appendix~\ref{a:related}.
\fi

\medskip
\noindent
\textbf{Distributed Property Testing.}~
As noted previously, the study of distributed property testing was initiated by Censor-Hillel et al.~\cite{CFSV16}.
In particular, they designed and analyzed distributed property testing algorithms for: triangle-freeness, cycle-freeness, and bipartiteness. As noted previously, they also proved a logarithmic lower bound for the latter two properties.
Fraigniaud et al.~\cite{FRST16} studied distributed property testing of excluded subgraphs of size $4$ and $5$.


Since the appearance of the above papers, there was a 
fruitful line of research in distributed property testing for various properties, mainly focusing on properties of whether a graph excludes a sub-graph, e.g., triangle-freeness, cycle-freeness, subgraphs of constant size, tree-freeness, clique freeness~\cite{FO17,FGO17c,FMORT17,ELM17,DISC17,FGO17a}.
In~\cite{FY17} the problem of testing the conductance of the input graph was studied, and a two-sided error tester is given.

Brakerski and Patt-Shamir~\cite{BP11} considered a related problem in the distributed setting.
They show how to find a subset of vertices that is $\eps$-close to being a clique if there is
a subset that is $\eps^3$-close to being a clique.

\medskip
\noindent
\textbf{Distributed Algorithms with a promise in the \congest\ model.}~
There is a large variety of distributed algorithms that
work under a promise that the graph is planar (similarly, excludes a fixed minor, has bounded arboricity, and more).
See e.g.~\cite{BE10,LW10,LPW13,ASS16,GH16,GHb16,HIZ16,HIZb16,GL17,GP17,HLZ18}.
\newcommand{\ptestingrel}{
\ifnum\podc=0
\ifnum\cready=0
\paragraph{Centralized Property Testing.}
\else

\medskip
\noindent
\textbf{Centralized Property Testing.}~
\fi
\fi
Most works in property testing on $H$-minor freeness (and related properties) focus on two-sided error testers.
The problem of testing general $H$-minor freeness was studied by Benjamini, Schramm and Shapira in \cite{BSS08}.
They showed that every minor-closed property of bounded degree graphs is testable with two-sided error with query complexity which is independent of the size of the graph.
Subsequent work~\cite{HKNO09,LR15} improved the dependence on the proximity parameter, $\eps$.
Yoshida and Ito~\cite{YI15} provided a tester with two-sided error for outerplanarity whose query complexity is $\poly(1/\eps)$.
Eden et~al.~\cite{ELR18} studied tolerant testing of bounded arboricity in the general graph model and showed almost tight bounds in terms of the dependence on $n$ and $m$.

Czumaj et~al.~\cite{CGRSSS14} studied property testing of $H$-minor freeness with one-sided error. They proved that for $H$ which is a forest, $H$-minor freeness can be tested in query complexity which is independent on $n$. For any $H$ that contains a cycle they showed a lower bound of $\Omega(\sqrt{n})$ queries. They also provided almost matching upper bounds for any $H$ which is a simple cycle.
Recently, Fichtenberger et~al.~\cite{FLVW17} provided a $\tilde{O}(n^{2/3})$-query tester for outerplanarity, and other properties that can be characterized by forbidden minors, with one-sided error.

Czumaj et~al.~\cite{CMOS11} studied testing bipartitness under the promise the input graph is planar. While in general, the query complexity of testing bipartitness is $\Omega(\sqrt{n})$, even for bounded degree graphs, they showed that under the promise of planarity, bipartiteness can be tested in time which is independent of $n$, even if the maximum degree is unbounded.
}
\ifnum\podc=0
\ptestingrel
\fi 

%% file: alg.tex
\section{The algorithm for testing planarity}\label{sec:plan-test}
In this section we establish the following theorem.
\begin{theorem}\label{thm:test-plan}
There exists a distributed one-sided error property testing algorithm for planarity that runs in
$O(\log(n)\cdot\poly(1/\eps))$ rounds in the \congest\ model.  
\end{theorem}

As noted in the introduction (Section~\ref{subsec-intro-high-level}), our algorithm works in two stages:
a {\em Partition stage\/} and a {\em Planarity testing stage\/}.
In what follows we describe and analyze each in detail, where we refer to the first as Stage~I and the second as Stage~II. %
\ifnum\podc=1
All missing details for this section can be found in Appendix~\ref{a:partition}.
\fi


\input{partition}

 \subsection{A detailed description and analysis of Stage~II}\label{subsec:stage2}
  Assume Stage~I completes successfully. By Claim~\ref{clm:stage-I-correct1}, this always holds when $G$ is planar, and by the definition of planarity, each subgraph $G^j$ is planar. On the other hand,
  if $G$ is $\eps$-far from being planar, then by Claim~\ref{clm:stage-I-correct1}, at least one subgraph $G^{j*}$ is $(\eps/2)$-far from being planar. That is, if for each $j \in [k]$ we let $m(G^j)$ denote the number of edges in $G^j$, then the number of edges that need to be removed from $G^{j*}$ in order to make it planar is at least
  $(\eps/2)m(G^{j*})$.

 \subsubsection{Preliminary preprocessing rounds}
 Stage~II begins with several preliminary rounds of {\em basic information gathering}, where we build on
  Corollary~\ref{cor:stage-I-correct2} and Lemma~\ref{lem:span-tree}.
  Specifically, we use the fact  that each $G^j$ is connected, has diameter $\poly(1/\eps)$ and has a
  designated root node, $r^j$, that is known to all nodes in $G^j$.
\ifnum\podc=1
\begin{compactitem}
\else
\begin{itemize}
\fi
 \item In the first $\poly(1/\eps)$ rounds, for each $j\in [k]$, the nodes in each $G^j$ construct
 a BFS tree, rooted at $r^j$ and denoted $T_B^j$.
 \ifnum\podc=0
 This is done simply as follows.
 The root $r^j$ sends a message $(r^j,r^j,0)$ to all its neighbors (indicating that it belongs to level $0$
 in the BFS tree rooted at $r^j$). Once a node $u$ in $G^j$  receives a message $(r^j,v,s)$ from a neighbor $v$, $u$ notifies  $v$ that it is $v$'s child in the tree, and sends a message $(r^j,u,s+1)$ to all its neighbors (this is of course done only once, upon receiving this first such message).
 \fi
 When this process terminates, each node in $G^j$ knows which of its incident edges is incident to its parent in $T_B^j$, which edges are incident to its children in $T_B^j$, and which are edges in $G^j$ that do not belong to $T_B^j$ (which we refer to as {\em non-tree edges\/}).
 Each  edge in $G^j$ is assigned to its higher-level endpoint (breaking ties by ids in the case of edges with both endpoints in the same level).

\item  In the next $\poly(1/\eps)$ rounds, for each $j\in [k]$, the root $r^j$ obtains the number of nodes $n(G^j)$ in $G^j$ and the number of edges, $m(G^j)$. This is done simply by sending the corresponding information up the tree $T_B^j$.
 \ifnum\podc=0
    Namely, to obtain $n(G^j)$, each  node sends its parent the number of nodes in its subtree (once it obtains the number of nodes in the subtrees of its children). Similarly, to obtain $m(G^j)$, each node sends its parent the number of edges assigned to nodes in its subtree.
 \fi
 Once $n(G^j)$ and $m(G^j)$ are computed by $r^j$, it can broadcast them (down $T_B^j$) to all nodes in $G^j$ (in another  $\poly(1/\eps)$ rounds).
\ifnum\podc=1
\end{compactitem}
\else
\end{itemize}
\fi

 \subsubsection{Planarity testing within each $G^j$}
If $m(G^j) > 3n(G^j) - 6$, then $r^j$ rejects. Otherwise, the distributed planar embedding algorithm of Ghaffari and Haeupler~\cite{GH16} is executed on $G^j$. Recall that if $G^j$ is planar, then this algorithm computes for each node a circular 
ordering of its incident edges (known as a {\em combinatorial embedding}), such that there exists a planar (geometric) embedding of $G^j$ that is {\em consistent} with all edge orderings.
If some node $v$ does not obtain an ordering of its incident edges in $G^j$ (within the allotted number of rounds: $O(D(G^j) + \min(\log(n(G^j)),D(G^j))$ where $D(G^j)$ is the diameter of $G^j$)), then it rejects.


As noted in the introduction, it was communicated to us by one of the authors of~\cite{GH16} that their algorithm can be modified so as to detect if any $G^j$ is not planar~\cite{H-p-com}.
For the sake of a self-contained presentation, we rely on version of the algorithm as provided in~\cite{GH16} (which works under the promise that $G^j$ is planar). Therefore, it
remains to show how to verify (efficiently) that the ordering of edges incident to each node is indeed consistent with a planar embedding. To this end, we first introduce several notations and definitions.

Let $\tau = \{\tau_u\}_{u\in V}$ denote the ordering of edges in $G^j$ (as computed by the algorithm of Ghaffari and Haeupler~\cite{GH16}), which is with respect to the clockwise order of the embedding, and let $\tilde{\tau}$ be the same ordering of edges as $\tau$ but with respect to the counter-clockwise order.
In particular for each $u \in V$, and
for every three edges $(u,v_1),(u,v_2),(u,v_3)$, $\tilde{\tau}_u$ indicates whether $(u,v_2)$ is between $(u,v_1)$ and $(u,v_3)$ in counter-clockwise order.\footnote{Note that for any subgraph $H$ of $G^j$, the ordering $\tilde{\tau}$ defines a combinatorial embedding of $H$ as well.} 

Using the ordering $\tilde{\tau}$, together with the BFS tree $T_B^j$, each node in $G^j$ associates a label  with each of its incident edges (and in particular those incident to its children in $T_B^j$).
Specifically, $r^j$ arbitrarily labels one of its incident edges $e$ by `1', and the remaining edges are labeled
 consecutively according to their order with respect to $e$. For each other node $u$ in the tree, if the circular 
 counter-clockwise order of its incident edges is $(u,v_1)$, $(u,v_2)$, $\dots,$ $(u,v_k)$ where $v_1$ is its parent in the BFS tree,
 then $u$ labels  $(u,v_j)$ by $j$, and we denote this by $\ell(u,v_j)=j$ (indeed each edge has two labels, one from each of its endpoint, so that $\ell(v_j,u)$ may differ from $\ell(u,v_j)$).\footnote{The reason that we use the counter-clockwise order rather than the clockwise order is so that the edges in the tree will be in the standard, left-to-right order.}

 This labeling of edges is then used to induce a labeling on the nodes of $G^j$ in the natural manner: the label of a node $u$, denoted $\ell(u)$ is the concatenation of the labels of the edges on the path in $T_B^j$ from $r^j$ down to the node, where for each edge we use the label associated by the parent node. This labeling can be computed in $\poly(1/\eps)$ rounds, by distributing the label information down the tree, starting from $r^j$. This labeling of nodes defines a lexicographic order on the nodes.\footnote{That is,
 for two (different) nodes $u$ and $v$, let $\ell(u) = \sigma_1,\dots,\sigma_p$ and $\ell(v) = \sigma'_1,\dots,\sigma'_q$, where without loss of generality, $p \leq q$ (and for $u = r^j$ we have $p=0$).
  Let $i$ be the maximum index such that $\sigma_1,\dots,\sigma_i = \sigma'_1,\dots,\sigma'_i$, where if no such index exists, then $i = 0$. If $i=p$, then $\ell(u) < \ell(v)$. Otherwise, $\ell(u) < \ell(v)$ if $\sigma_{i+1}<\sigma'_{i+1}$, and $\ell(u) > \ell(v)$ otherwise.}

 We next introduce several notions regarding cycles and violating edges.
 See Figure~\ref{fig:def7def8} for an illustration of the notions introduced in
 Definitions~\ref{def:cycle-inside-out} and~\ref{def:inside-out}.

\begin{figure*}
    \centering
    \begin{subfigure}[b]{0.45\textwidth}
        \centering
        \includegraphics[width=\textwidth]{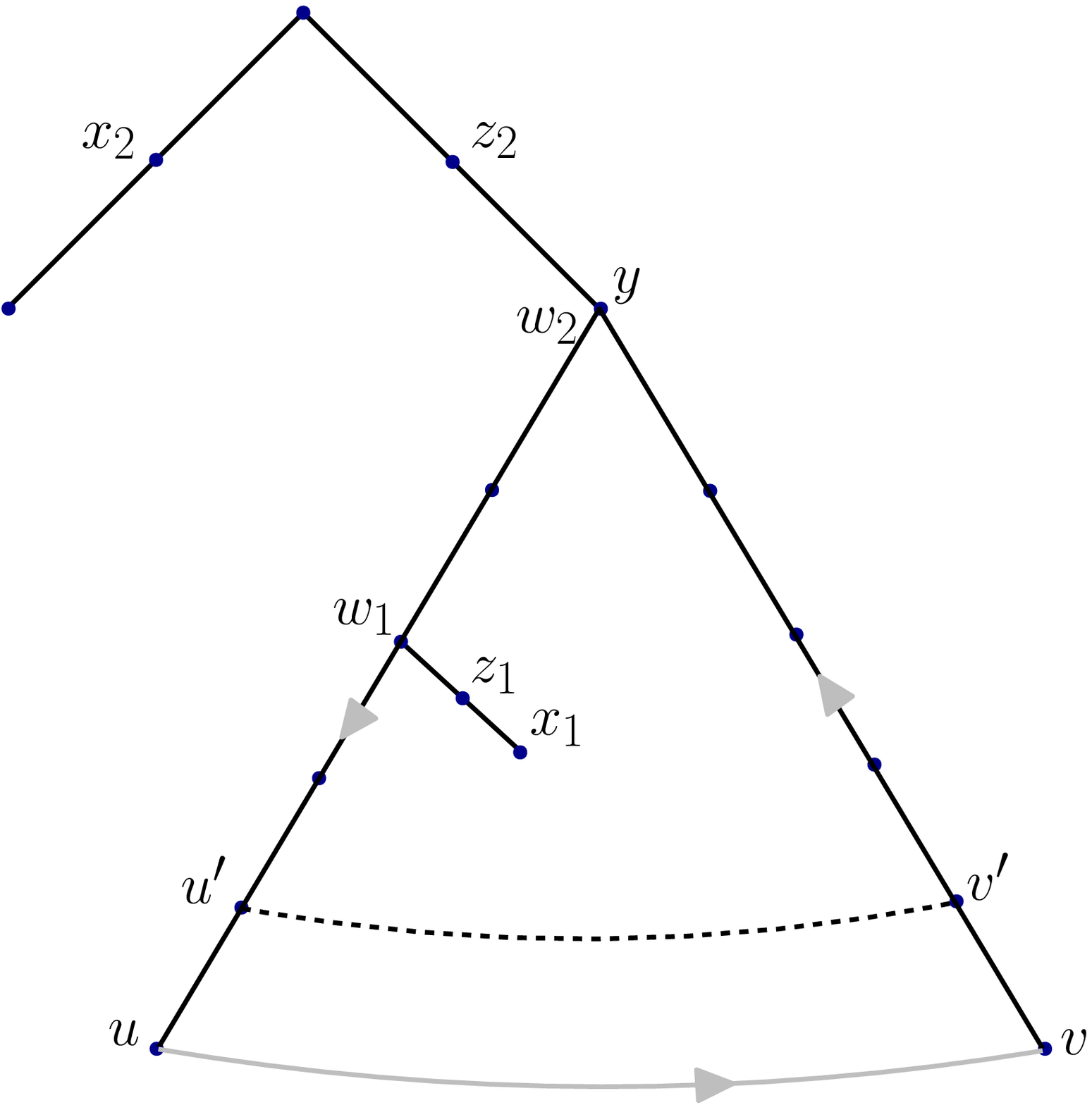}
        \caption{}
        \label{fig:def7def8}
    \end{subfigure}
    \quad\quad\quad
    \begin{subfigure}[b]{0.45\textwidth}
        \centering
        \includegraphics[width=\textwidth]{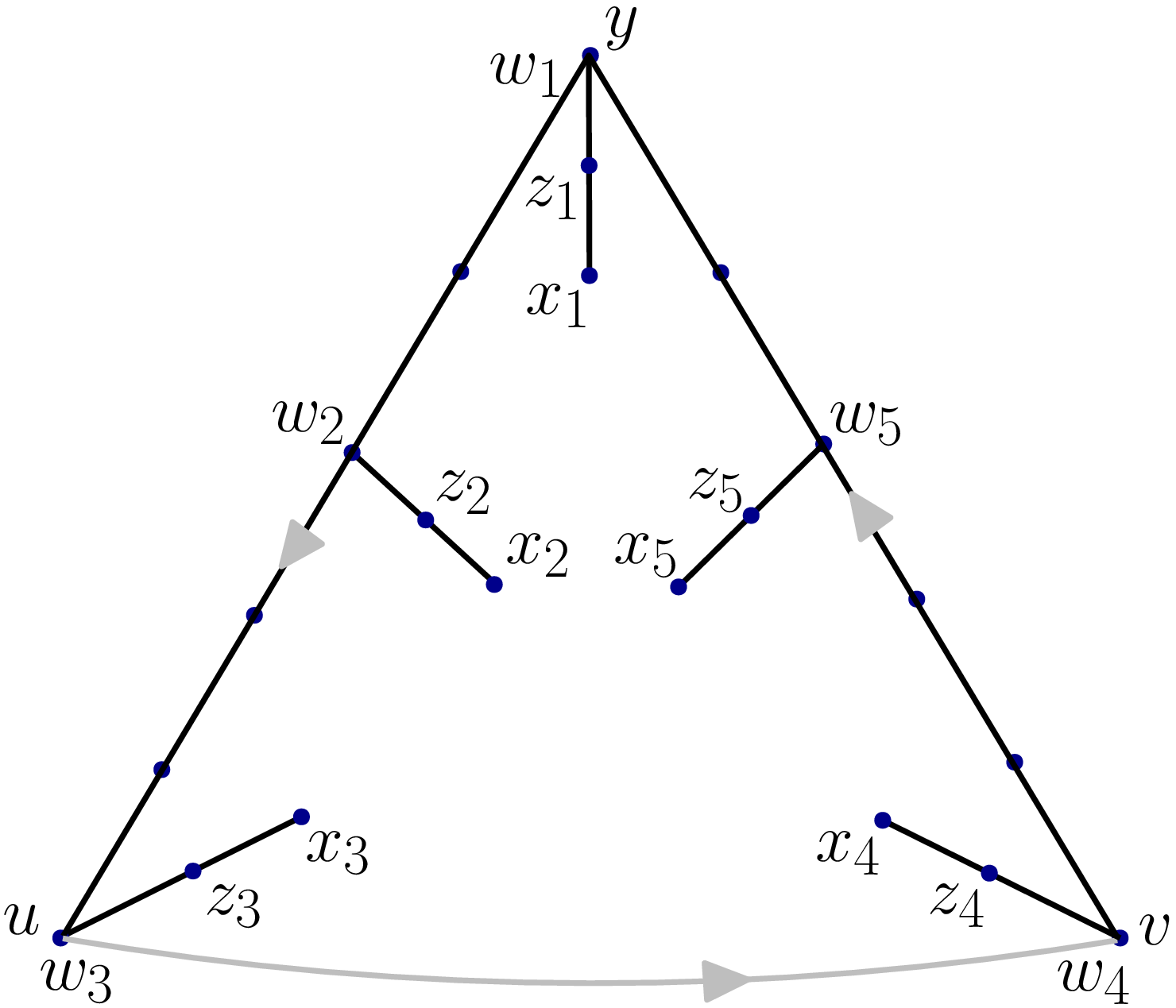}
        \caption{}
        \label{fig:incases}
    \end{subfigure}
    \newline
    \begin{subfigure}[b]{0.45\textwidth}
        \centering
        \includegraphics[width=\textwidth]{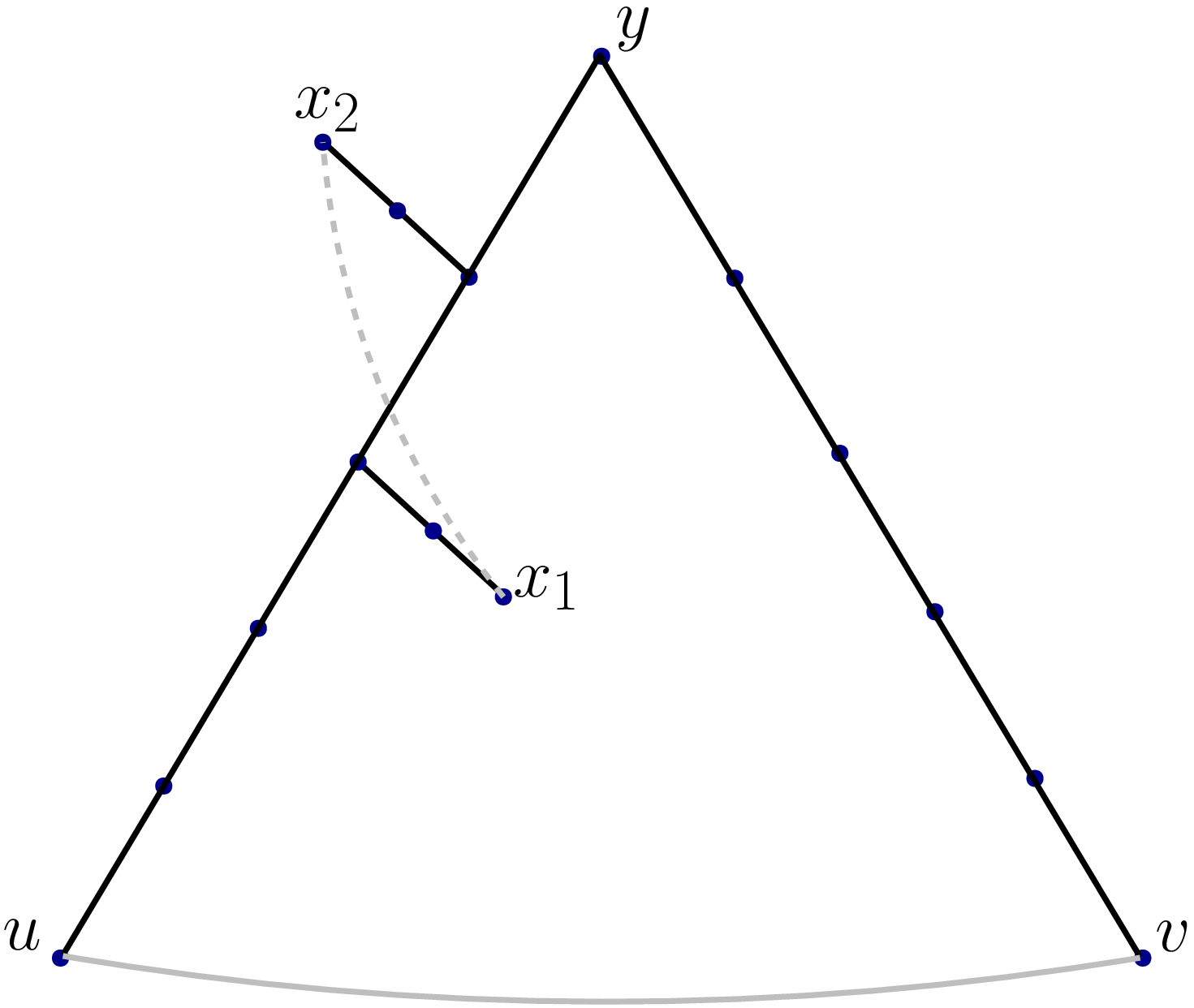}
        \caption{}
        \label{fig:violate}
    \end{subfigure}
    \quad\quad\quad
    \begin{subfigure}[b]{0.45\textwidth}
        \centering
        \includegraphics[width=\textwidth]{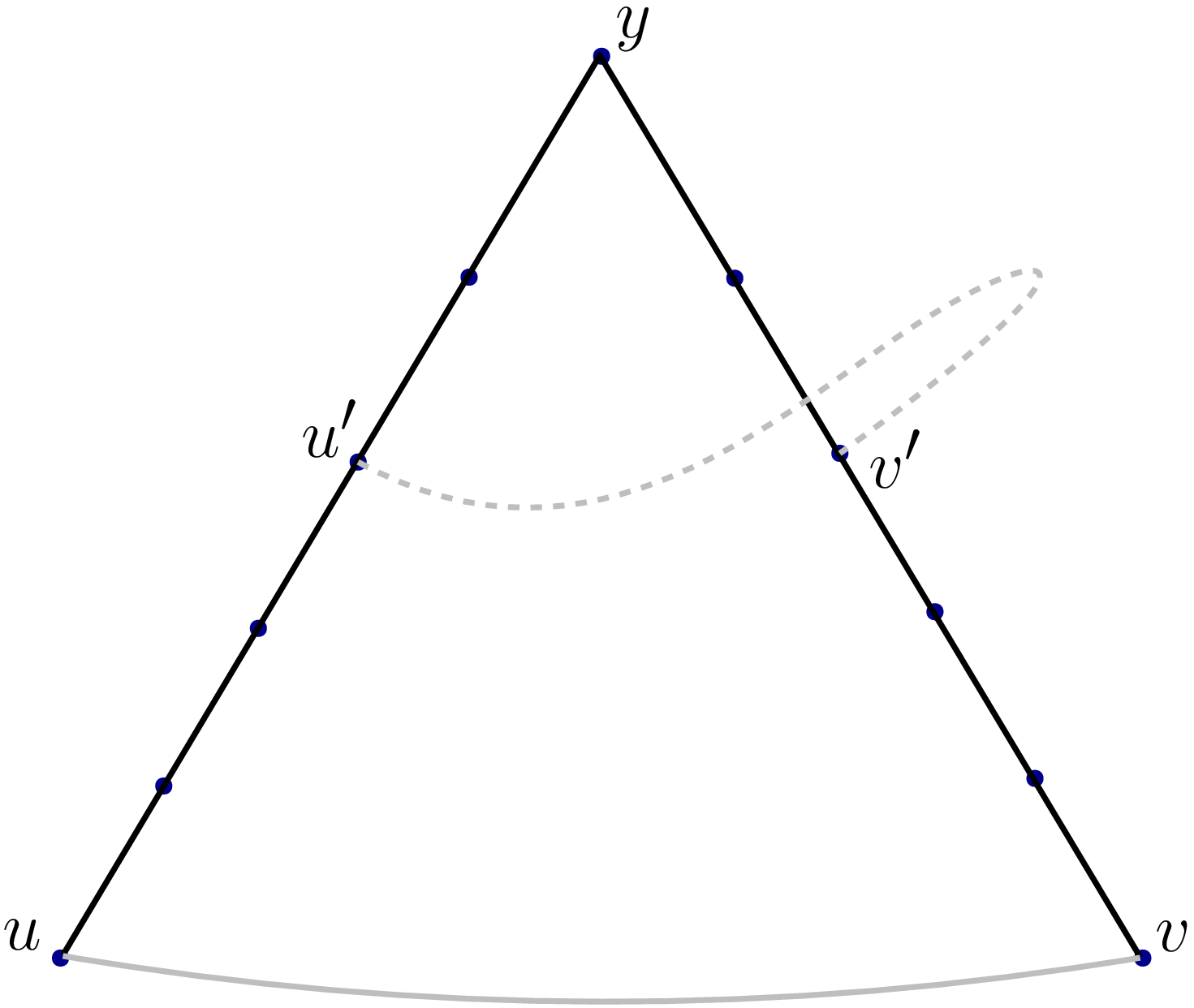}
        \caption{}
        \label{fig:violate2}
    \end{subfigure}
    \caption{In all subfigures, the BFS tree edges are depicted as black edges. The non-tree edge $(u,v)$ is depicted by a gray edge - forming the simple cycle $C(u,v)$. The vertices of $C(u,v)$ are ordered as depicted by the red arrows, i.e., $(y,\ldots,w_1,\ldots, u,v,\ldots, y)$.
    \\
    Two  cases are considered in Sub-figure~\ref{fig:def7def8}: (1)~$x_1$ is {\sf inside}, and (2)~$x_2$ is {\sf outside}. Moreover, the edge $(u',v')$ is also {\sf inside $C(u,v)$ w.r.t. $u'$ and $v'$}. The shortest paths from $C(u,v)$ to these $x$'s are along the BFS tree.
    \\
    In Sub-figure~\ref{fig:incases} we give examples of nodes $x_i$ that are {\sf inside} $C(u,v)$: (1)~$x_1$ satisfies: $\ell(u) < \ell(x_1)<\ell(v)$ and  $x_1$ is not a descendant of $u$ (similarly for $x_2$ and $x_5$); (2)~$x_3$ satisfies: (a) $\ell(u) < \ell(x_3)<\ell(v)$, (b) it is a descendant of $u$ and (c) $\ell(u,z_3)>\ell(u,v)$;  (3)~$x_4$ is a descendant of $v$ and $\ell(v,z_4)<\ell(v,u)$.
    \\
    In Sub-figures~\ref{fig:violate} and~\ref{fig:violate2} we give examples of violating edges: (1)~in Sub-figure~\ref{fig:violate}, $x_1 \in I(u,v)$ while $x_2 \in O(u,v)$, hence the edge $(x_1,x_2)$ is a violating edge. In Sub-figure~\ref{fig:violate2} the edge $(u',v')$ is {\sf inside $C(u,v)$ w.r.t. $u'$} but is {\sf outside $C(u,v)$ w.r.t. $v'$}, hence it is a violating edge.}
    \label{fig:cycle_vio}
\end{figure*}

  \begin{definition}\label{def:cycle-inside-out}
Let $H$ be a subgraph of $G^j$, and let $C$ be a cycle in $H$, where we consider a fixed ordering
of the vertices on the cycle: $C = (x_1,\dots,x_k,x_1)$.
 For a vertex $x_i$ on the cycle and an edge $(x_i,y)$, we say that  $(x_i,y)$ is
 {\sf inside $C$ with respect to $x_i$ (and $\tilde{\tau}$)}, if $(x_i,y)$ is between $(x_i,x_{i+1})$ and
$(x_i,x_{i-1})$ in the (counter-clockwise) ordering of edges incident to $x_i$ as defined by $\tilde{\tau}$ (where for $i=1$, $x_{i-1}=x_k$ and similarly for $i=k$, $x_{i+1}=x_1$).
Otherwise, $(x_i,y)$ is  {\sf outside $C$} (with respect to $x_i$ and $\tilde{\tau}$).
\end{definition}

\begin{definition}\label{def:inside-out}
  Let $(u,v)$ be a  non-tree edge in $G^j$ (with respect to the BFS tree $T_B^j$) where $\ell(u) < \ell(v)$.
  Let $y$ be the least common ancestor of $u$ and $v$ in $T_B^j$, and let $C(u,v)$ be the simple cycle consisting of
  (in this order), the path from $y$ to $u$ in $T_B^j$, the edge $(u,v)$,  and the path from $v$ to $y$ in $T_B^j$.

   Let $x$ be a node that does not belong to $C(u,v)$, and consider the node on $C$, $w$, for which the
length of the path from $u$ to $w$ in $T_B^j$ in minimized. 
   We say that $x$  is {\sf inside $C(u,v)$} if the first edge $(w,z)$ on this path
   is inside $C(u,v)$ with respect to  $w$ and $\tilde{\tau}$ (as defined in Definition~\ref{def:cycle-inside-out}). Otherwise $x$ is {\sf outside $C(u,v)$}.
\end{definition}

Observe that by Definition~\ref{def:inside-out} and the definition of the labeling $\ell$,
 a node $x$ is inside $C(u,v)$ if and only if one of the following conditions holds (see Figure~\ref{fig:incases}).
  \begin{itemize}
  \item $\ell(u) < \ell(x) < \ell(v)$ and  $x$ is not a descendant of $u$;
  \item $\ell(u) < \ell(x) < \ell(v)$, $x$ is a descendant of $u$ and
  $\ell(u,z) > \ell(u,v)$ where $z$ is the ancestor of $x$ that is a child of $u$.
   \item $x$ is a descendant of $v$ and $\ell(v,z)< \ell(v,u)$ where $z$ is the ancestor of $x$ that is a child of $v$.
  \end{itemize}

Using Definitions~\ref{def:cycle-inside-out} and~\ref{def:inside-out} we define violations between edges.
\begin{definition}\label{def:violate-edges}
  Let $(u,v)$  and $(u',v')$ be two non-tree edges in $G^j$ (with respect to the BFS tree $T_B^j$).
   We say that $(u',v')$ {\sf is in violation} with respect to $(u,v)$ if one of the following conditions hold.
     \begin{itemize}
     \item $u'$ and $v'$ are both on $C(u,v)$, and the edge $(u',v')$ is on different sides of $C(u,v)$ with respect to its two endpoints (see Figure~\ref{fig:violate2}).
  \item  $u'$ is inside
  $C(u,v)$ and $v'$ is outside of $C(u,v)$ (see Figure~\ref{fig:violate}). 
  \item
  $u'$ is inside (outside) of $C(u,v)$, 
  $v'$ is on $C(u,v)$, and the edge $(u',v')$ is
  outside (respectively, inside) $C(u,v)$. (An illustration for this case is very similar to that shown in Figure~\ref{fig:violate}, and is hence omitted.)
   \end{itemize}
   We say that  $(u,v)$ is a {\sf violating\/} edge, if there exists at least one non-tree edge $(u',v')$
 that is violating with respect to $(u,v)$.
\end{definition}

    \begin{ourclaim}\label{clm:planar-non-intersect}
  If $G^j$ is planar, then there are no violating edges in $G^j$.
  \end{ourclaim}

  \begin{proof}
  Consider a planar embedding of $G^j$ that is consistent with the ordering $\tilde{\tau}$.
  Recall that $\tilde{\tau}$ is the output of the algorithm of
  Ghaffari and Haeupler~\cite{GH16}, and that by the correctness of the algorithm, such an embedding exists.
   For each non-tree edge $(u,v)$ 
   consider the cycle $C(u,v)$ as defined in Definition~\ref{def:inside-out}.
   Let $u'$ and $v'$ be two nodes in $G^j$.
   If $u'$ and $v'$ are both on $C(u,v)$ and there is an edge between them, then either $(u',v')$ is inside
   $C(u,v)$ both with respect to $u'$ and with respect to $v'$ or it is outside $C(u,v)$ with respect to both.
   Otherwise there is no planar embedding of $C(u,v) + (u',v')$ that is consistent with $\tilde{\tau}$.
   Similarly, if neither $u'$ nor $v'$ belongs to $C(u,v)$,
   suppose that $u'$ is defined to be
  {\em inside\/} $C(u,v)$ according to Definition~\ref{def:inside-out} and $v'$ is defined to be {\em outside\/}, or vice versa.
   Since the planar embedding is consistent with $\tilde{\tau}$, the nodes $u'$ and $v'$
  must be embedded on different sides of $C(u,v)$, and hence there cannot be an edge between $u'$ and $v'$.
  Similar arguments hold for the case that one of the two nodes is on $C(u,v)$ and the other is inside/outside
  $C(u,v)$.
  \end{proof}

\medskip
Recall that $T_B^j$ is a BFS tree defined over $G^j$ and that violations are defined with respect to a labeling $\ell$ that is 
induced by $\tilde{\tau}$. For a subgraph $H$ of $G^j$, we say that $H$ is {\em connected by $T_B^j$} if
for every two nodes $u$ and $v$ in $H$, there is a path in $H$ consisting only of edges belonging to $T_B^j$.
We shall use the following notation in our proof that if $G^j$ does not contain any violating edges, then it is planar
(Claim~\ref{clm:no-violating-planar}).
\begin{definition}\label{def:in-out}  
Let $H$ be a subgraph of $G^j$  that is connected by $T_B^j$.
For any  edge $(u,v)$ in $H$ that does not belong to the tree $T_B^j$ where $\ell(u) < \ell(v)$, let $I_H(u,v)$ and $O_H(u,v)$  denote the subset of nodes in $H$ that are inside and outside $C(u,v)$, respectively (according to Definition~\ref{def:inside-out}).
\begin{itemize}
\item Let $E_H^I(u,v)$ denote the set of edges that are either incident to $I_H(u,v)$ or edges for which both endpoints are on $C(u,v)$ and are inside $C(u,v)$ with respect to both endpoints.
\item Similarly, let $E_H^O(u,v)$ denote the set of edges that are either incident to $O_H(u,v)$ or edges for which both endpoints are on $C(u,v)$ and are outside $C(u,v)$ with respect to both endpoints.
\item Let $S_H^I(u,v)$ denote the subgraph consisting of $C(u,v)$ and $E_I(u,v)$.
\item Similarly, let $S_H^O(u,v)$ denote the subgraph consisting of $C(u,v)$ and $E_O(u,v)$.
\end{itemize}
\end{definition}

We will also use the following lemma. In order to prove the lemma we apply a slight variant of Lemma~7.2 and Corollary~7.1 in~\cite{Even}.
Since the details are very similar to those appearing in~\cite{Even}, they are deferred to the appendix.

\begin{lemma}\label{cor:cycle}
Let $C(u, v)$ be a cycle in $H$ as defined in Definition~\ref{def:inside-out}. If both $S_H^I(u,v)$ and $S_H^O(u,v)$ (see Definition~\ref{def:in-out}) are planar and have a planar embedding that is consistent with $\tilde{\tau}$, then $H$ is planar and has a planar embedding that is consistent with $\tilde{\tau}$ as well.
\end{lemma}

We are now ready to prove Claim~\ref{clm:no-violating-planar}, stated next.

  \begin{ourclaim}
  \label{clm:no-violating-planar}
  If $G^j$ does not contain any violating edges, then it is planar.
  \end{ourclaim}

\begin{proof}
We show that for any subgraph $H$ of $G^j$ that is connected by  $T_B^j$, if there are no violating edges in $H$,
then there exists a planar embedding of $H$ that is consistent with $\tilde{\tau}$.


We prove this claim by induction on the pair $(t,h)$, where $t$ is the number
of non-tree edges in $H$, and $h$ is the total number of edges.
The base cases are $t=0$ and $t = 1$  (for any $h$), for which it is easy to see that $H$ is always planar.
(Observe that if $h \leq 3$, then $t \leq 1$.)

For the induction step, consider a subgraph $H$ (connected by the edges of $T_B^j$) with $t\geq 2$ non-tree edges and $h$ edges.
Note that for any non-tree edge $(u, v)$,  $H$ is the union of $S_H^I(u,v)$ and $S_H^O(u,v)$, since there are no edges between nodes in $I_H(u,v)$ and $O_H(u,v)$.
We consider two cases.
The first case is that there exists a non-tree edge, $(u, v)$, such that both $E_I(u,v)$ and $E_O(u,v)$ are non-empty.
Observe that both $S_H^I(u,v)$ and $S_H^O(u,v)$ are connected by edges of $T_B^j$. Since there are no violating edges in $H$, there are also no violating edges in each of these two subgraphs.
We can therefore apply the induction hypothesis (since the number of edges in each subgraph is strictly smaller than in $H$), and infer that both $S_H^I(u,v)$ and $S_H^O(u,v)$ have planar embeddings, and furthermore, that each of these embeddings is consistent with $\tilde{\tau}$.
By Lemma~\ref{cor:cycle}, the claim follows.

The second case is that there exists a non-tree edge, $(u, v)$, such that $E_I(u,v)$ is empty.
In this case we remove $(u, v)$ and consider a planar embedding that is consistent with $\tilde{\tau}$ of the resulting graph $H'$. By the induction hypothesis such an embedding of $H'$ exists. Now, we claim that it is possible to
add $(u, v)$ 
to this embedding and obtain a planar embedding of $H$ that is consistent with $\tilde{\tau}$.
To verify this, observe that the inside of $C(u, v)$ with respect to $\tilde{\tau}$ is empty and therefore it is possible to add the edge $(u, v)$, in a manner that is consistent with $\tilde{\tau}$, without crossing any edges of $H'$.

If both cases do not occur, then it is implied that either there are no cycles (namely, $t=0$) or that for all non-tree edge, $(u, v)$, $E_O(u,v)$ is empty. The latter implies that there is a single non-tree edge in the graph (i.e., $t=1$). This completes the proof of the claim.
\end{proof}

\medskip
  As a corollary of Claim~\ref{clm:no-violating-planar} we get:
  \begin{corollary}
  If $G^j$ is $\gamma$-far from 
  planarity, then there exist at least $\gamma\cdot m(G^j)$ violating edges in   $G^j$.
  \end{corollary}

\medskip
  Given Claims~\ref{clm:no-violating-planar} and~\ref{clm:planar-non-intersect}, the algorithm proceeds as follows. First $r^j$ broadcasts the labels of 
  $s = \Theta(\log n/\eps)$
  non-tree edges of $G^j$ that are selected uniformly, independently, at random. Such a selection can be performed in $\log n \cdot \poly(1/\eps)$ rounds.
 In particular, each node can decide independently for each of the non-tree edges assigned to it whether it is selected (by flipping a coin with bias $\Theta((\log n/\eps)/\tilde{m}^j)$ for each of these edges, where
 $\tilde{m}^j$ is the total number of non-tree edges in $G^j$). The selected edges (i.e., pairs of node labels) are sent up the tree, where if the number of selected edges is significantly larger than the expected number, then the algorithm fails (this happens with probability $1/\poly(n)$).
Once  $r^j$ obtaines such a sample of non-tree edges, it broadcasts the labels of these edges to all nodes in the tree. Each node in the tree can now check whether any of the non-tree edges assigned to it
is in violation with
any one of the sampled edges, and reject based on such
a violation.

\ifnum\podc=0
\medskip
We have thus completed establishing Theorem~\ref{thm:test-plan}.  
\fi

%% file: partition.tex

 \subsection{A detailed description and analysis of Stage~I}\label{subsec:stage1}


Recall from the description in the introduction (Section~\ref{subsec-intro-high-level}) that
 Stage~I consists of $t= O(\log(1/\eps))$ phases.  Each phase is associated with a partition of the nodes,
 where in the initial partition each node is a singleton part, and in general, each part is connected.
 We denote by $\calP_i = (P_i^1,\dots,P_i^{k_i})$ the partition associated with Phase $i$.
 Each partition $\calP_i$ defines an auxiliary weighted graph, denoted $\calG_i$.
 Each part $P_i^j$ corresponds to a node in $\calG_i$, denoted $\ver(P_i^j)$, and the weight of
 an edge $(\ver(P_i^j),\ver(P_i^{j'}))$ is the number of edges with one endpoint in $P_i^j$ and the
 other in $P_i^{j'}$. Each part in $\calP_{i+1}$ is a union of several parts in $\calP_i$.

 Each phase consists of two steps: A {\em forest decomposition step\/} and a {\em merging step\/}.
 We first describe these steps in terms of the auxiliary graphs $\{\calG_i\}_{i=1}^{t+1}$ (so that each
 part $P_i^j$ is viewed as a single node, which may be denoted by $\calv$ or $\calu$ rather than
 $\ver(P_i^j)$).
 We then explain how they are emulated  on $G$.

\paragraph{Sending message up and down trees.}
 Both in $G$ and in the auxiliary graphs $\{\calG_i\}_{i=1}^{t+1}$, we consider distributing information on trees (where each node
 knows which of its incident edges is incident to its parent in the tree and which of these edges are incident to its children). For such a tree $T$, let $r(T)$ denote its root.
 When we say that $r(T)$ sends a message {\em down the tree\/}, we mean that $r(T)$ sends the message to its children in $T$ and they send it to their children, and so on, until the message reaches all nodes in $T$.
 In some cases the message may be augmented/modifed as it goes down the tree.
When we say that the a node $v \in T$ sends a message {\em up the tree\/}, we mean that $v$ sends the message to its parent in $T$, which in turn send it to its parent, and so on, until it reaches the root $r(T)$.
Here too a message may be modified as it goes up the tree. In particular,
if several nodes simultaneously send different messages up the tree, then this may cause congestion, and we
explain how this is addressed whenever it arises.

 \subsubsection{The forest decomposition step on $\calG_i$}\label{subsubsec:forest-Gi}
This step correspond to the forest decomposition algorithm of Barenboim and Elkin~\cite{BE10}.
 Their algorithm works under the promise that the underlying
graph $\calG_i$ has arboricity at most  $\alpha$ where in our case $\alpha$ is set to 3 (the bound on the arboricity of planar graphs). The algorithm ignores the weights on the edges of $\calG_i$, and
proceeds as follows. Initially all
 nodes in $\calG_i$ are {\em active}.
For $s = \Theta(\log n)$ rounds, each active node $\calu$ does the following.
If $\calu$  has at most $3\alpha$ active neighbors (in the current round),
then $\calu$ sends a message to all its neighbors that it becomes {\em inactive\/} in the next round.
As shown in~\cite{BE10} (and is not hard to verify), if $\calG_i$ has arboricity at most $\alpha$,
then in each round a constant fraction of the nodes become inactive.
Since the number of rounds is $\Theta(\log n)$ (and the number of nodes in $\calG_i$ is at most $n$),
if $\calG_i$ has arboricity at most $\alpha$,
 all the nodes are inactive by the termination of the algorithm.
In other words, if some node remains active after $s$ rounds, then this serves as evidence that
$\calG_i$ has arboricity larger than $\alpha$.

If the process terminates successfully (i.e., all nodes become inactive), then
it is possible  to define a forest decomposition into at most $3\alpha$ forests. That is, it is possible to orient the edges of $\calG_i$ so that each node in $\calG_i$ has at most $3\alpha$ outgoing edges, one to each of its parents in the different forests (where no cycles are formed). Specifically, consider a node $\calu$ that becomes inactive in round $\ell$, and let $\calv_1,\dots,\calv_d$, for $d\leq 3\alpha$, be its active neighbors at the start of round $\ell$. For each $\calv_q$ such that $\calv_q$ remains active in round $\ell+1$, we orient the edge $\{\calu,\calv_q\}$ from $\calu$ to $\calv_q$. For each $\calv_q$ such that
$\calv_q$ also becomes inactive in round $\ell$, we orient the edge $\{\calu,\calv_q\}$
from the node with the smaller id to the node with the larger id (so that no directed cycles are formed).
We emphasize that each node knows the orientation of its incident edges.

\subsubsection{The merging step: from $\calG_i$ to $\calG_{i+1}$} \label{subsubsec:merge-Gi}
Assuming the forest decomposition step completed successfully (and hence each node in
$\calG_i$ has at most $3\alpha$ outgoing edges pointing to its parents in the forest decomposition),
the merging step consists of the following  sub-steps.
Sub-step~\ref{substep:subtrees} is as in~\cite{CHW08}.
\ifnum\podc=1
\begin{compactenum}
\else
\begin{enumerate}
\fi
\item  \label{substep:heaviest}
Each node $\calu$ in $\calG_i$ selects its outgoing edge $(\calu,\calv)$ that has the highest weight.
Let $\calF_i$ denote the forest induced by the selected edges.
\item \label{substep:subtrees} Select a set of ``shallow'' subtrees of $\calF_i$, denoted $\calT_i$,
as follows:
  \begin{enumerate}
\item\label{substep:cole-vishkin}
Obtain a coloring $\chi$ of $\calF_i$ using colors in $\{1,2,3\}$
by running the distributed algorithm of Cole and Vishkin~\cite{CV86}, and Goldberg, Plotkin and Shannon~\cite{GPS88}.
\item\label{substep:mark} Mark the edges of $\calF_i$ according to the rules defined next
(where if a node marks an incident edge, it notifies the other endpoint).
\ifnum\podc=1
\begin{compactitem}
\else
\begin{itemize}
\fi
\item
For each node $\calu$
such that $\chi(\calu)=1$,  $\calu$ marks its outgoing edge (assuming such exists) if the weight of this edge is greater or equal to the sum of the weights of all its incoming edges. Otherwise, $\calu$  marks all its incoming edges.
\item For each node $\calu$ such that $\chi(\calu)=2$,
 $\calu$ marks its outgoing edge if the other endpoint is colored $3$ and the weight of this edge is
 greater or equal to the sum of the weights of all its incoming edges whose other endpoint is colored $3$. Otherwise it marks all these incoming edges.
\ifnum\podc=1
\end{compactitem}
\else
\end{itemize}
\fi
\item Let $\calT_i$ be the set of subtrees induced by the marked edges.
  \end{enumerate}
  \item\label{substep:stars}  For each subtree $T \in \calT_i$, let $w_0(T)$ denote the total weight of edges that go from an even level in $T$ up to an odd level (referred to as ``even edges''), and let $w_1(T)$ denote the total weight of the remaining edges in $T$ (referred to as ``odd edges''). The root of $T$, denoted $r(T)$, obtains $w_0(T)$ and $w_1(T)$ (by sending a message down the tree so that each node learns its level, and receiving message sent up the tree in
      which weights of even/odd edges are summed).
      If $w_0(T) \geq w_1(T)$, then $r(T)$ sends the message `0' down the tree, and otherwise it sends `1'.
  \item \label{substep:star-contract} If the message sent down the tree is `0', then all even edges are contracted, and otherwise all odd edges are contracted.
\ifnum\podc=1
\end{compactenum}
\else
\end{enumerate}
\fi
Observe that each node in $\calG_{i+1}$ corresponds to a star subgraph in $\calG_i$.

\medskip
Let $w(\calG_i)$ denote the total weight of edges in $\calG_i$ (and similarly define
$w(\calF_i)$ and $w(\calT_i)$).
\begin{ourclaim}\label{clm:merge}
The merging step runs in $O(\log^* n)$ rounds (on $\calG_i$) and
$w(\calG_{i+1}) \leq \left(1 - \frac{1}{12\alpha}\right)\cdot w(\calG_i)$. 
\end{ourclaim}

\newcommand{\mergestepclaim}{
\begin{proof}
First observe that by the definition of $\calF_i$ (in Sub-step~\ref{substep:heaviest}),
$w(\calF_i) \geq w(\calG_i)/3\alpha$. By~\cite{CV86}, the number of rounds performed in
Sub-step~\ref{substep:cole-vishkin} is $O(\log^* n)$ (and hence $O(\log^* n)$ bounds
the number of rounds performed in all of Sub-step~\ref{substep:subtrees}).
By the analysis presented in~\cite[Section 2]{CHW08},
$w(\calT_i) \geq w(\calF_i)/2$, and the height of each tree $T$ in $\calT_i$ is at most 10.
Therefore, Sub-step~\ref{substep:stars} runs in a constant number of rounds (and the same is true of
Sub-step~\ref{substep:star-contract}). Finally, by the choice of which edges to contract
(in Sub-step~\ref{substep:stars}),
the weight of the contracted edges is at least $w(\calT_i)/2$. The claim follows.
\end{proof}
}
\ifnum\podc=0
\mergestepclaim
\fi

\subsubsection{Successful completion of Stage~I}\label{subsubsec:success-I}
Before turning to the emulation of Stage~I on $G$, we introduce one definition and two claims, whose correctness does not depends on the details of the emulation.

\begin{definition}\label{def:stage1-success}
 We say that Stage~I {\sf completes successfully} if  the forest decomposition step in each phase
 terminates with no remaining active node.
\end{definition}
  \begin{ourclaim}
 \label{clm:stage-I-correct1}
 If $G$ is planar, then Stage~I always completes successfully. If $G$ is $\eps$-far from being planar, then
 either Stage~I does not complete successfully, or
 $w(\calG_{t+1}) \leq\eps m/2$.
  \end{ourclaim}

\newcommand{\stagecomplete}{
 \begin{proof}
 The first part of the claim follows immediately from the fact that the arboricity of planar graphs is at most $3$ and that any  minor of a planar graph is planar.
  The second part of the claim follows from
 the fact that $w(\calG_1) = m$,
 Claim~\ref{clm:merge}, and the setting of the number of phases $t = \Theta(\log(1/\eps))$.
 \end{proof}
}
\ifnum\podc=0
\stagecomplete
\fi

    \begin{ourclaim}
  \label{clm:stage-I-correct2}
 For each phase $i$ and part $P_i^j$, the subgraph induced by $P_i^j$ is connected and has
 diameter at most $4^i$.
 \end{ourclaim}

\newcommand{\connecdiam}{
 \begin{proof}
 We prove the claim by induction on $i$. The claim trivially holds for $i=1$.
 To establish the induction step,
consider a single merging step in which the nodes $\ver(P_i^{j_1}),\dots,\ver(P_i^{j_s})$ in $\calG_i$ all merge with $\ver(P_i^{j_0})$ (that is, the edges $\left\{\left(\ver(P_i^{j_q}),\ver(P_i^{j_0})\right)\right\}_{q=1}^s$
were contracted).
Clearly, the subgraph induced by $\bigcup_{q=0}^s P_i^{j_q}$ (which corresponds to a part
in $\calP_{i+1}$) is connected.
As for the diameter of this subgraph, by the induction hypothesis,
it is at most $3\cdot 4^i +2 \leq 4^{i+1}$.
\end{proof}
}
\ifnum\podc=0
\connecdiam
\fi

\medskip
Conditioned on Stage~I completing successfully, let
 $\calP = (P^1,\dots,P^k)$ denote the final partition (i.e., $\calP = \calP_{t+1}$).
 For each $j \in [k]$, let $G^j = G(P^j)$ denote the subgraph induced by $P^j$.
 As a corollary of Claim~\ref{clm:stage-I-correct2} we get:
 \begin{corollary}\label{cor:stage-I-correct2}
 Each $G^j$ is connected and has diameter $\poly(1/\eps)$.
 \end{corollary}

\subsubsection{Preliminaries for the emulation} 

For each phase $i \in [t]$ and part $P_i^j\in \calP_i$, let $G_i^j$ denote the subgraph induced by $P_i^j$.
We say that a node $u\in P_i^j$ is a {\em boundary\/} node, if at least one of its neighbors belongs to
a part $P_i^{j'}$ for $j' \neq j$.

In Section~\ref{subsec:emulate-merge} we establish the following lemma.
\begin{lemma}\label{lem:span-tree}
For every $i\in [t]$ and every $P_i^j \in \calP_i$, there is a unique root node $r_i^j \in P_i^j$, such that each node in $P_i^j$ knows the identity of $r_i^j$. Furthermore, there is a spanning tree of $G_i^j$, rooted at $r_i^j$ and denoted $T_i^j$, for which the following holds. Each node $u$ in $G_i^j$ knows which of its incident edges is incident to its parent in $T_i^j$ and which of these edges is incident to its children.
\end{lemma}

\ifnum\podc=1
Due to space constraints, we provide the emulation details only for the forest decomposition
step (in Section~\ref{subsec:emulate-forest}). Details for the emulation of the merging step can be found in
Appendix~\ref{subsec:emulate-merge}.
\fi

%

\subsubsection{Emulating the forest decomposition step}\label{subsec:emulate-forest}

\def\Active{\mbox{\rm `Active'}}

We assume that at the start of each phase $i$, the conditions in Lemma~\ref{lem:span-tree} hold.
We refer to each round in the forest decomposition algorithm described in Section~\ref{subsubsec:forest-Gi}
(on $\calG_i$) as a {\em super-round\/}.
%
In the forest decomposition step of each phase $i$, we have $s=O(\log n)$ super-rounds in which each active node $\ver(P_i^j)$ in $\calG_i$ needs to determine if it is still active in the next super-round and to send a corresponding message to its neighbors.
Each super-round is emulated by several rounds (on $G$), as described next.

The root of $P_i^j$, $r_i^j$, plays the role of $\ver(P_i^j)$ as follows.
For each super-round $\ell$, at the start of which $\ver(P_i^j)$ is active,
if $r_i^j$ determines in the course of this super-round that $\ver(P_i^j)$ should remain active in the next super-round ($\ell+1$), then it sends a message $(\Active, r_i^j)$ down the tree $T_i^j$.
Each boundary node in $P_i^j$ also sends this message to its neighbors outside of $P_i^j$.
The process by which $r_i^j$ determines in super-round $\ell$ whether $\ver(P_i^j)$ remains active or not is defined as follows
(where this process is also executed one super-round after $\ver(P_i^j)$ becomes inactive so that
$r_i^j$ can learn which neighbors of $\ver(P_i^j)$ also became inactive in super-round $\ell$).

At the beginning of each super-round $\ell$, each boundary node $u$ in $P_i^j$ that received
in the previous super-round
messages of the form $(\Active, r_i^{j'})$ for $j' \neq j$, 
does the following.
If $u$ received more than $3\alpha$ such messages with {\em distinct\/} root ids, then it sends a message \Active\ up the tree
(meaning that $\ver(P_i^j)$ should remain active since it has more than $3\alpha$ active neighbors).
Otherwise, for each $r_i^{j'}$ such that $u$ received a message $(\Active, r_i^{j'})$,
$u$ sends a message $(\Active,r_i^{j'}, x)$ to its parent, where $x$ indicates how many messages $(\Active, r_i^{j'})$ it received.
These messages go up the tree, where if a node $u$ receives the message \Active, then this is the single message it passes on. If $u$ did not receive \Active\, but it received more than $3\alpha$ messages $(\Active,r_i^{j'},x)$ with distinct root ids, then it also sends \Active\ up the tree. Otherwise, for each $r_i^{j'}$, let the
messages $u$ received with this root id be $(\Active,r_i^{j'},x_1),\dots,(\Active,r_i^{j'},x_q)$. Then $u$
sends its parent a message $\left(\Active,r_i^{j'},\sum_{p=1}^q x_p\right)$.
Finally, if $r_i^j$ received the message \Active\ or if it received more than $3\alpha$ messages
with distinct root ids, then
$r_i^j$ determines that $\ver(P_i^j)$ remains active in the  next super-round. Otherwise, it determines
that $\ver(P_i^j)$ becomes inactive. In the latter case, not only that the out-edges of $\ver(P_i^j)$ can be determined by $r_i^j$, their weights can be determined as well (this is the role of the third parameter in the messages going up the tree).

\smallskip
The total number of rounds (on $G$) sufficient for emulating a single super-round on $\calG_i$ is hence
upper bounded by the maximum diameter of parts in $\calP_i$, which by Claim~\ref{clm:stage-I-correct2}
is $\poly(1/\eps)$ (times $3\alpha$, which is a constant).

\smallskip
If after all $O(\log(n))$ super rounds there is some $r_i^j$ such that $\ver(P_i^j)$ is still active, then $r_i^j$ outputs {\sf reject} (implying that Stage~I did not complete successfully).

\def\EmulateMerge{
\ifnum\podc=1
\subsection{Emulating the merging step}\label{subsec:emulate-merge}
\else
\subsubsection{Emulating the merging step}\label{subsec:emulate-merge}
\fi
In this subsection we explain how to emulate all sub-steps in the merging step, and
establish Lemma~\ref{lem:span-tree} by induction on $i$. The base of the induction, $i=1$
is trivial, since in $\calP_1$ each node belongs to a singleton part, and in one round each node learns the identity of all its neighbors.

\paragraph{Determining the heaviest out-edge (Sub-step~\ref{substep:heaviest}).}
For a part $P_i^j$, consider the super-round $\ell$ in the forest decomposition algorithm in which $\ver(P_i^j)$ becomes inactive.
By the end of this super-round, $r_i^j$ obtained the ids of the roots, $r_i^{j'_1},\dots,r_i^{j'_q}$
where  $q \leq 3\alpha$, such that
$\ver(P_i^{j'_1}),\dots,\ver(P_i^{j'_q})$ are the active neighbors of $\ver(P_i^j)$ in $\calG_i$ at the start of
super-round $\ell$. Furthermore, for each of the corresponding parts $P_i^{j'_p}$,
the root $r_i^j$
obtained the number of edges between $P_i^j$ and $P_i^{j'_p}$, so that it knows the weight
of $(\ver(P_i^j),\ver(P_i^{j'_p}))$ in $\calG_i$.

 It remains to determine which of these edges is an outgoing edge of
$\ver(P_i^j)$. If $\ell$ is not the final round, then in the next round, $r_i^j$ learns which
of nodes $\ver(P_i^{j'_1}),\dots,\ver(P_i^{j'_q})$ remained active in round $\ell+1$. For each
such node $\ver(P_i^{j'_p})$, the edge $(\ver(P_i^j),\ver(P_i^{j'_p}))$ is an outgoing edge of $\ver(P_i^j)$,
and for each $\ver(P_i^{j'_p})$ that also became inactive in super-round $\ell$, the direction of the edge is determined by the ids of $r_i^j$ and $r_i^{j'_p}$. If $\ell$ is the last round, then either some
node in $\calG_i$ remained active, causing the corresponding root node in $G$ to reject, or all
nodes in $\calG_i$ became inactive, so that edge directions are determined by root ids.

For the sake of the following sub-steps, it will be convenient to designate, for each
selected outgoing edge $(\ver(P_i^j),\ver(P_i^{j'}))$, a single edge $(u,v)$ in $G$ such that
$u \in P_i^j$ and $v \in P_i^{j'}$. To this end, let $(\ver(P_i^j),\ver(P_i^{h(i,j)}))$ be the
heaviest outgoing edge of $\ver(P_i^j)$. The root, $r_i^j$ send a message with the id of
$r_i^{h(i,j)}$ down the tree. Each node $u\in P_i^j$ that has a neighbor in $ P_i^{h(i,j)})$
sends its id up the tree, where if a node receives more than one ``candidate'' node id from its children, then it sends the minimum id among them. In this manner, $r_i^j$ obtains the id of a single node $u_i^j \in P_i^j$
that has a neighbor, $v_i^j$, in $P_i^{h(i,j)}$, and it can notify $u_i^j$ that it ``in charge'' of
of the outgoing edge of $P_i^j$ (by sending an appropriate message down the tree).

\paragraph{Selecting (marking) shallow subtrees (Sub-step~\ref{substep:subtrees}).}
In order to emulate this sub-step it is first necessary to emulate the coloring algorithm of Cole and Vishkin~\cite{CV86}, and Goldberg, Plotkin and Shannon~\cite{GPS88}. The important observation is  that in this algorithm, whenever a node $\ver(P_i^j)$ in $\calF_i$ sends a message to its children in $\calF_i$ (where a message is always of size $O(\log n)$ bits), it sends the same message.
Hence, this can be emulated by simply sending (broadcasting) this message from $r_i^j$ down the tree $T_i^j$. Once the message reaches the boundary nodes of $P_i^j$ they also send it to their neighbors in $G$, and the message can go up the trees
of the parts corresponding to the children of $\ver(P_i^j)$. Sending a message from $\ver(P_i^j)$
to its parent in $\calF_i$ is similar (and even simpler, since there is a single parent).

In order to emulate the marking of edges, each root $r_i^j$ needs to gather information regarding the number of
edges between $P_i^j$ and parts $P_i^{j'}$ such that $j = h(i,j')$, for the different color classes. This information can be easily gathered by sending appropriate messages up the tree, and summing edge counts that correspond to the same color.

\paragraph{Deciding if to contract even or odd edges in each tree (Sub-step~\ref{substep:stars}).}
The emulation of this part is also simple. For each tree $T \in \calT_i$, first messages should be sent down $T$, so that each $\ver(P_i^j)$ can learn its level. Each such message from $\ver(P_i^j)$ to its children in $T$,
$\ver(P_i^{j_1}),\dots,\ver(P_i^{j_q})$ is emulated by sending a message down $T_i^j$ from $r_i^j$, and then
up the trees $T_i^{j_p}$. In a similar manner messages are sent up the tree $T$, summing up the weights of
even and odd edges, and then the bit `0' or `1' is sent down $T$.

\paragraph{Contracting edges (Sub-step~\ref{substep:star-contract}).}
Once a root $r_i^j$ corresponding to a node $\ver(P_i^j)$ learns that the edge
$(\ver(P_i^j),\ver(P_i^{h(i,j)}))$ should be contracted, it
sends a message down the tree $T_i^j$ notifying all nodes that $r_i^{h(i,j)}$ is their
new root. When this message reaches $u_i^j$ (the node in charge of the
edge $(\ver(P_i^j),\ver(P_i^{h(i,j)}))$), it makes $v_i^j$ (its neighbor in $P_i^{h(i,j)}$) its
parent, and sends a message up $T_i^j$ that each edge on the path to $r_i^j$ should flip its orientation.
The induction step for Lemma~\ref{lem:span-tree}, follows.

\paragraph{Emulation cost.}
The total number of rounds (on $G$) sufficient for emulating the merging step (from $\calG_i$ to
$\calG_{i+1}$) is hence upper bounded by $O(\log^*(n))$ (by Claim~\ref{clm:merge})
 times the maximum diameter of parts in $\calP_i$, which by Claim~\ref{clm:stage-I-correct2} is $\poly(1/\eps)$.
} 

\ifnum\podc=0
\EmulateMerge
\fi

%% file: lb.tex
\def\wG{{\widetilde{G}}}
\newcommand{\lbproof}{
\section{A lower bound}\label{sec:lb}
Recall that for a fixed graph $H$,
$H$ is a {\em minor\/} of a graph $G$ if $H$ is isomorphic to a graph that can be obtained
by zero or more edge contractions on a subgraph of $G$.
We say that a graph $G$ is {\em $H$-minor free\/} (or {\em excludes $H$
as a minor\/}) if $H$ is not a minor of $G$.
For a family ${\mathcal H}$ of (constant-size) graphs, we say that a graph $G$ is
{\em ${\mathcal H}$-minor free\/} if it is $H$-minor free for every $H \in {\mathcal H}$.
In particular, planar graphs are $\{K_{3,3},K_5\}$-minor free.

In this section we 
establish the following theorem, which extends a result of
Censor-Hillel et al.~\cite[Theorem 7.3]{CFSV16}
for $K_3$-minor freeness (cycle-freeness).\footnote{To be precise,
the graphs in the lower-bound construction of Censor-Hillel have a constant degree, while the
graphs in our lower bound construction do not necessarily have a constant degree.
However, we can easily modify
the construction so that the graphs have a constant degree,
in the same manner as
in~\cite[Theorem 7.3]{CFSV16}.}
\begin{theorem}\label{thm:lb}
Let ${\mathcal H}$ be a fixed family of constant-size graphs where at least one $H \in {\mathcal H}$ contains a cycle.
Any distributed one-sided error  algorithm for testing ${\mathcal H}$-minor freeness
must run in $\Omega(\log n)$ rounds (for constant $\eps$).
\end{theorem}
Our proof of Theorem~\ref{thm:lb} is very similar to the proof of Theorem 7.3 in~\cite{CFSV16}. We also build on a lower bound proof of Czumaj et al.~\cite{CGRSSS14} for one-sided error testing of minor-freeness.
Similarly to~\cite{CFSV16}, we use the probabilistic method to establish that for any constant $k$ and any number of nodes $n$,
there exist graphs $G$ over $n$ nodes for which the following hold: (1) $G$ is $\eps$-far
from being $K_k$-minor free for $\eps = \eps(k)$; (2) $G$ contains no cycles of length $\log(n)/c$ for a sufficiently large constant $c = c(k)$.
Theorem~\ref{thm:lb} directly follows by setting $k$ to be the minimum size of $H \in {\mathcal H}$ that contains a cycle and observing that for any one-sided error algorithm that runs in less than $\log(n)/c$  rounds, when executed on $G$, all nodes must accept.

In order to construct such graphs, we first select a graph $\wG$ distributed according to
$\calG(n,p)$ for $p=\Theta(1/n)$, and prove that it is far from being
 $K_k$-free with high probability. We then show that by removing a relatively small number of edges, the resulting graph, $G$, has no short cycles, and remains far from being $K_k$-free.
 \ifnum\cready=1
The  next claim is similar to~\cite[Claim 6.2]{CGRSSS14} and its proof can be found in the full version of this paper.
 \fi

\begin{ourclaim}\label{clm:far}
Let $\wG$ be a graph selected according to $\calG(n,p)$ for $p= 1000 k^2 /n$. With probability
$1- 2^{\Omega(n)}$, the graph
$\wG$ has at most $2000k^2 n$ edges and is $\eps$-far from $K_k$-minor freeness for $\eps = 1/(50 k^2)$.
\end{ourclaim}

\ifnum\cready=0
\begin{proof}
Since the expected number of edges in $\wG$, denoted $m(G)$, is $p \cdot {n\choose 2}<p n^2 = 1000 k^2 n$, by a multiplicative Chernoff bound, the probability that $m(G) > 2000 k^2 n$ is at most $e^{-pn^2/3} = 2^{-\Omega(n)}$.
From this point on we condition on the event that $m(G) \leq 2000 k^2 n$.

We  say that $\wG$ is {\em well connected\/} if
for every two disjoint subsets $C_1$ and $C_2$ of nodes such that
$|C_1|,|C_2| \geq n/3k$, the number of edges with one endpoint in $C_1$ and the other in $C_2$
is greater than $\eps m$.
We next establish the following subclaim: With probability $1- 2^{\Omega(n)}$,
the graph  $\wG$ is well connected.
For any two subsets $C_1$ and $C_2$ of nodes such that
$|C_1|,|C_2| \geq n/3k$, the expected number of edges between them
is at least $p (n/3k)^2 \geq 100 n$. Once again by a multiplicative Chernoff bound,
the probability that there are less than $50 n$ edges between the two sets is
at most $e^{-5n}$. The number of such pairs of subsets is upper bounded by $3^n$, and so the
probability that for some such pair of subsets there are less than $50 n$ edges between them, is
upper bounded by $2^{-\Omega(n)}$. Setting $\eps = 1/(50 k^2)$, the subclaim follows.
From this point we also condition on the event that $\wG$ is well connected.

The remainder of the argument follows~\cite[Proof of Claim 6.2]{CGRSSS14}.
Consider an arbitrary partition of the nodes in $\wG$ into $k$ equal size subsets, $U_1,\dots,U_k$,
and let $\wG_i$ be the subgraph induced by $U_i$. We claim that each $\wG_i$ contains a connected component of size at least $n/3k$. To verify this, let $W_i^1,\dots,W_i^t$ be the connected components of $\wG_i$.
Assume, contrary to the claim, that each connected component contains less than $n/3k$ nodes.
But this means that there exists a subset of indices $J \subset [t]$ such that both $W_i = \bigcup_{j\in J}W_i^j$ and
$W'_i = \bigcup_{j\in [t]\setminus J}W_i^j$ contain at least $n/3k$ nodes each. But since $\wG$ is well connected, there must be
an edge between some node in $W_i$ and some node in $W'_i$, and we get a contradiction. We thus have, for each part $U_i$, a connected component, $W_i^{j(i)}$ of size at least $n/3k$.
Using once again the assumption that $\wG$ is well connected, we get that
for each pair $(W_i^{j(i)},W_{i'}^{j(i')})$, there are more than $\eps m$ edges between
 $W_i^{j(i)}$ and  $W_{i'}^{j(i')}$.
This implies that $\wG$
is $\eps$-far from being $K_k$-minor free.
\end{proof}
\fi

The next claim  uses the same argument as in~\cite[Lemma 7.7]{CFSV16}.
\begin{ourclaim}\label{clm:cycles}
Let $\wG$ be as defined in Claim~\ref{clm:far} and let $G$ be a graph resulting from $\wG$ by removing
a single edge from each cycle in $\wG$ whose length is less than $\log(n)/c(k)$,
where $c(k) = \Theta(\log k)$.
With probability at least $1/2 - 2^{-\Omega(n)}$, the graph $G$ is $\eps$-far from $K_k$-minor
freeness for $\eps = 1/(100 k^2)$.
\end{ourclaim}

\begin{proof}
Let $S$ be a fixed set of $\ell$ nodes. The probability (over the choice of $G$) that there is a cycle over
$S$ is at most $\ell! \cdot p^\ell$. Therefore,
the expected number of cycles of length at most $\ell$ is upper bounded by
${n\choose \ell} \cdot \ell!\cdot p^\ell < (1000 k^2)^\ell$, and with probability at least $1/2$
it is at most twice this number.
If we set $\ell = \log(n)/\log(1000 k^2)$
then
by Claim~\ref{clm:far} and a union bound over all ``bad'' events, Claim~\ref{clm:cycles} follows.
\end{proof}
}
\ifnum\podc=0
\lbproof
\fi

%% file: second-partition.tex
\section{A partitioning algorithm for minor-free graphs and applications}\label{sec:minor-free}
In what follows, when we use the term ``a distributed partitioning algorithm'', we mean an algorithm that gives the following guarantee. Upon completion, there is a partition
 $\calP = (P^1,\dots,P^k)$ of the nodes such that for each $j \in [k]$, the subgraph induced by $P^j$ is connected, and there is a designated node $r^j \in P^j$ such that all nodes in $P^j$ know the id of $r^j$.
We first note that Stage I of our testing algorithm (described in Section~\ref{subsec:stage1})
implies the next theorem.

\sloppy
\begin{theorem}\label{thm:determ}
There exists a deterministic distributed partitioning algorithm in the \congest\ model for which the following holds.
For an edge-cut parameter $\eps\in (0,1)$, the algorithm runs in $O(\poly(1/\eps)\log n)$ rounds, the diameter of each part is $\poly(1/\eps)$, and if $G$ is minor-free,
then the total number of edges between parts is at most $\eps n$.
\end{theorem}

We show that by modifying the algorithm referred to in Theorem~\ref{thm:determ}, we obtain a
tradeoff between the round complexity and the success probability, as stated next.
\ifnum\podc=1
The proof appears in Appendix~\ref{a:parti2}.
\fi
\begin{theorem}\label{thm:second}
There exists a distributed partitioning algorithm in the \congest\ model for which the following holds.
For an edge-cut parameter $\eps\in (0,1)$ and a confidence parameter $\delta\in (0,1)$, the algorithm runs in $O(\poly(1/\eps)(\log(1/\delta)+\log^* n))$ rounds, the diameter of each part is $\poly(1/\eps)$, and if $G$ is minor-free,
then with probability at least $1-\delta$, the total number of edges between parts is at most $\eps n$.
\end{theorem}

\begin{remark}
If one is willing to settle for constant success probability (i.e., constant $\delta$), then the round complexity of Theorem~\ref{thm:second} can be improved to be only $\poly(1/\eps)$.
\end{remark}
\newcommand{\secondparti}{
In what follows we prove Theorem~\ref{thm:second}.

Recall that in the algorithm for testing planarity described in Section~\ref{sec:plan-test}, the source of the dependence on $\log(n)$ in the round complexity was due to the forest decomposition step in each phase of the partition stage.  
This step served to verify that each $\calG_i$ has constant arboricity, as well as to allow for each node in $\calG_i$ to select its heaviest outgoing edge in the corresponding forest decomposition (when the arboricity is bounded as required).

If however, there is a promise that $G$ is $H$-minor free, for any fixed $H$ of constant size $h$, then the arboricity of every $\calG_i$ is upper bounded by a constant $c(h)$.
Therefore, there is no need to perform this arboricity verification step.
 Furthermore, as we show below, instead of selecting the heaviest outgoing edge in a forest decomposition, it suffices to select a random incident edge, where the probability to select an edge is a function of its weight.

As discussed above, the algorithm referred to in Theorem~\ref{thm:second} is a modified version of Stage~I of the planarity testing algorithm. It too runs in $\log(1/\eps)$ phases, where in Phase $i$ it coarsens the partition $\calP_i$ and obtains the partition $\calP_{i+1}$ (initially, $\calP_1$ is the partition into singleton parts).
The first difference is that the forest decomposition step is not executed. The second difference is
in the choice of an incident edge for each node in $\calG_i$.
We next describe how the choice of an incident edge is performed as well as the merge decision. In Section~\ref{subsec:random-edge} we explain how this choice and decision are emulated
on $G$. Once a decision to merge is made, the emulation of the merge is performed as described in Section~\ref{subsec:emulate-merge}.

Recall that the algorithm in Section~\ref{subsubsec:merge-Gi}, which runs on $\calG_i$, consists of $4$ Sub-steps.
Sub-step~\ref{substep:heaviest} is the only sub-step that is modified:
Instead of selecting the heaviest out-going edge, each node in $\calG_i$ randomly selects an edge as described next.
The resulting graph, i.e., the graph induced on the selected edges, is guaranteed to be a directed pseudo-forest: the edge selected by each node is its only out-edge, and if an edge is selected by both endpoints then it is oriented as the out-edge of the node of lower id.
We note that for Sub-steps~\ref{substep:cole-vishkin} and~\ref{substep:mark} we only rely on the fact that $\calF_i$ is a directed pseudo-forest.
In order to show that the marking process in Sub-step~\ref{substep:mark} results in a graph which is a forest we prove Claim~\ref{clm:stilltree}.
This claim is required for the correctness of Sub-step~\ref{substep:stars} of the algorithm.

\paragraph{Edge Selection.}
Let $\alpha$ be the arboricity of $\calG_i$ (which is a constant since $G$ is minor-free).
Each node in $\calG_i$ draws one of its incident edges with probability that is proportional to its weight.
Namely, for a node $\calu \in \calG_i$ and an edge $(\calu, \calv)$, the probability that $\calu$ draws
$(\calu, \calv)$ is $\frac{w(\calu, \calv)}{w(\calv)}$ where
$w(\calv) = \sum_{(\caly,\calv)\in E(\calG_i)} w(\caly, \calv)$.
This is repeated $s = \Theta(\log(1/\delta))$ times, and then each node selects the edge of maximum weight over the $s$ trials.
We call this {\em weighted-edge selection}.

\medskip\noindent
We prove the following lemma.
\begin{lemma}\label{lem:contr}
With probability at least $1-\delta$, the total weight of the edges selected in $\calG_i$ is at least $\frac{w(\calG_i)}{16\alpha}$.
\end{lemma}
\begin{proof}
Consider a forest decomposition of $\calG_{i}$ into $\alpha$ forests.
Orient the edges from children to parents so that the out-degree of each node is at most $\alpha$.
Let $w_{\rm out}(\calv)$ denote the weight of the out-going edges incident to $\calv$.
Observe that $w(\calG_{i}) = \sum_{\calv\in V(\calG_i)} w_{\rm out}(\calv) = \frac{1}{2} \sum_{\calv\in V(\calG_i)} w(\calv)$.
Let $U$ denote the set of nodes, $\calv$, such that $w_{\rm out}(\calv) \geq w(\calv)/4$.
Then,
\ifnum\jformat=1
\begin{eqnarray}
\sum_{\calv \in U} w_{\rm out}(\calv) & = & \sum_{\calv \in V} w_{\rm out}(\calv) - \sum_{\calv \notin U} w_{\rm out}(\calv) \nonumber \\ & \geq & w(\calG_{i}) - \sum_{\calv \notin U} w(\calv)/4 \;\;\geq\;\; w(\calG_{i})/2\ .\label{eq:weight}
\end{eqnarray}
\else
\ifnum\cready=0
\begin{equation}
\sum_{\calv \in U} w_{\rm out}(\calv)  = \sum_{\calv \in V} w_{\rm out}(\calv) - \sum_{\calv \notin U} w_{\rm out}(\calv) \geq w(\calG_{i}) - \sum_{\calv \notin U} w(\calv)/4 \geq w(\calG_{i})/2\ .\label{eq:weight}
\end{equation}
\else
\begin{eqnarray}
\sum_{\calv \in U} w_{\rm out}(\calv)  &=&
    \sum_{\calv \in V} w_{\rm out}(\calv) - \sum_{\calv \notin U} w_{\rm out}(\calv)  \nonumber \\
&\geq& w(\calG_{i}) - \sum_{\calv \notin U} w(\calv)/4
    \;\;\geq\;\; w(\calG_{i})/2\ .\label{eq:weight}
\end{eqnarray}
\fi
\fi 

Let $\calv \in U$ and $i\in [s]$. Define $w(\calv, i)$ to be the weight of the edge that $\calv$ drew in trial $i$.
Then $\Ex[w(\calv, i)] \geq w_{\rm out}(\calv)/(4\alpha)$.
To verify this, observe that with probability at least $1/4$, $\calv$ draws one of its out-edges, and conditioned on that, with probability at least $1/\alpha$, the heaviest out-edge is drawn.
By linearity of expectation and Equation~(\ref{eq:weight}), for each trial $i\in [s]$,  $\Ex\left[\sum_{\calv \in U}w(\calv, i)\right] \geq w(\calG_{i})/(8\alpha)$.
We claim that for every $i\in [s]$, with probability at least $1-1/(16\alpha-1)$,  $\sum_{\calv \in U}w(\calv, i) \geq w(\calG_{i})/(16\alpha)$.
Assume otherwise and obtain a contradiction:
\ifnum\jformat=0
$$\Ex\left[\sum_{\calv \in U}w(\calv, i)\right] < \frac{w(\calG_{i})}{16\alpha-1} + \left(1-\frac{1}{16\alpha-1}\right)\frac{w(\calG_{i})}{16\alpha} =  \frac{w(\calG_{i})}{8\alpha}\ .$$
\else
\begin{eqnarray}
\Ex\left[\sum_{\calv \in U}w(\calv, i)\right] & < & \frac{w(\calG_{i})}{16\alpha-1} + \left(1-\frac{1}{16\alpha-1}\right)\frac{w(\calG_{i})}{16\alpha} \nonumber\\ & = &  \frac{w(\calG_{i})}{8\alpha}\:.
\end{eqnarray}
\fi
Thus, the probability that in all $s$ trials we get that $\sum_{\calv \in U}w(\calv, i) < w(\calG_{i})/(16\alpha)$ is at most $(1/(16\alpha-1))^s$.
Since $s = \Theta(\log(1/\delta))$, we obtain that with probability at least $1-\delta$, there exists a trial $i \in [s]$ such that $\sum_{\calv \in U}w(\calv, i) \geq w(\calG_{i})/(16\alpha)$.
From the fact that $\sum_{\calv \in U}\max_i(w(\calv, i)) \geq \max_i\left(\sum_{\calv \in U}w(\calv, i)\right)$, we obtain the desired result.
\end{proof}

\begin{ourclaim}\label{clm:contract2}
With probability at least $1-\delta$, $w(\calG_{i+1})  \leq \left(1-\frac{1}{64\alpha}\right)\cdot w(\calG_i)$.
\end{ourclaim}

\begin{proof}
The proof follows from Lemma~\ref{lem:contr} and the same analysis as in Claim~\ref{clm:merge}.
\end{proof}

\begin{ourclaim}\label{clm:stilltree}
For an input graph that is a directed pseudo-forest, the graph resulting from the marking process in Sub-step~\ref{substep:mark} is a tree.
\end{ourclaim}
\begin{proof}
We first note that given a directed pseudo-forest, the only cycle that might exist in the graph has to be a directed cycle.
Assume towards contradiction that there exists a directed cycle in the marked graph.
We claim that it must contain a node that is colored by $1$.
To verify this, observe that any vertex that is colored by $2$, cannot have a marked outgoing edge and a marked incoming edge, such that for both edges, the other endpoint is colored by $3$.
But since every node that is colored by $1$ can only have either out-going edges or in-going edges, we reach a contradiction.
\end{proof}

\subsection{Emulation of the weighted-edge selection}\label{subsec:random-edge}
In what follows, when we say that an edge $(u,v) \in E(G)$ is incident to a part $P_i^j$, we
mean that $u \in P_i^j$ and $v \in P_i^{j'}$ for $j'\neq j$. In order to
draw an edge incident to $\ver(P_i^j)$ in $\calG_i$ with probability proportional to its weight
(i.e., emulate the drawing of edges in the weighted-edge selection),
we run a procedure for uniformly selecting an edge in $G$ incident to $P_i^j$.
If the selected edge in $G$ is $(u,v)$ where $v\in P_i^{j'}$, then
the corresponding drawn edge in $\calG_i$ is $(\ver(P_i^j), \ver(P_i^{j'}))$.


This uniform selection is implemented as follows. First, each node sends a message to all its neighbors with the id of the root of its part. Following this round, each node
 $u$ on the boundary of $P_i^j$, knows the set of incident edges  $(u,v)$ such that $v \notin P_i^j$.
 We denote this set by $E_{i,{\rm out}}(u)$  and let $d_{i,{\rm out}}(u) = |E_{i,{\rm out}}(u)|$.
Provided with this information, $u$ selects, uniformly at random, one edge $e \in E_{i,{\rm out}}(u)$
and sends its parent (in the tree $T_i^j$) the message $(e,d_{i,{\rm out}}(u))$.
In each consecutive round, if a node $v$ received the messages $(e_1,d_1),\dots,(e_s,d_s)$ from its
children, then it does the following. It sets $d = \sum_{p=1}^s d_p$, and selects one of the edges
$e_p$ with probability $d_p/d$. It then sends the message $(e_p,d)$ to its parent.
For the sake of consistency of the description, $r_i^j$ sends messages to itself.
At the end of this process (after $\poly(1/\eps)$ rounds), $r_i^j$ has a single edge, denoted
$e_i^j = (u_i^j,v_i^j)$ for $u_i^j \in P_i^j$ and $v_i^j \in P_i^{j'}$,
that is uniformly distributed among the edges incident to $P_i^j$.

\bigskip
Theorem~\ref{thm:second} now follows from Claim~\ref{clm:contract2}, and the fact that the total cost
of the weighted-edge selection is linear in the diameter of the parts $P_i^j$, which is $\poly(1/\eps)$, times the number of repetitions which is $O(\log (1/\delta))$.
}
\ifnum\podc=0
\secondparti
\fi
\subsection{Applications of the partitioning algorithm for minor-free graphs}\label{sec:apps}

As a corollary of Theorems~\ref{thm:determ} and~\ref{thm:second} we get the following.
\begin{corollary}\label{cor:cycle-free-bip}
There is a deterministic algorithm and a randomized algorithm in the \congest\ model
for testing the following properties on minor-free graphs: cycle-freeness and bipartiteness.
The deterministic algorithm runs in $O(\poly(1/\eps)\log n)$ rounds.
The randomized algorithm has one-sided error, runs in $O(\poly(1/\eps)(\log(1/\delta + \log^* n))$ rounds
and has success probability $1-\delta$.
\end{corollary}
We note that similar statements can be derived for any hereditary property that can either be verified
or (property) tested in a number of rounds that is polynomial in the diameter.

\newcommand{\coroapp}{
\medskip
\begin{proof}
For both properties, first the algorithm of Theorem~\ref{thm:determ} or Theorem~\ref{thm:second} is run with the edge-cut parameter set to slightly below $\eps$ (the distance parameter for property testing). Let $\calP = (P^1,\dots,P^k)$ be the resulting partition, and let $G^j$ denote
the subgraph induced by $P^j$.
By 
both theorems, if $G$ is $\eps$-far from having the property in question, then (with probability 1 or with probability
 $1-\delta$) at least one subgraph $G^j$ does not have the property (while if $G$ has the property, then every $G^j$ has the property). Therefore, it suffices to verify the property on each $G^j$. To this end,
an algorithm for finding a BFS tree is executed on each $G^j$. If the case of cycle freeness, each node now checks whether it has any incident non-tree edges in $G^j$, and in the case of bipartiteness it checks whether there is any such edge that closes an odd-length cycle.
\end{proof}
}
\ifnum\podc=1
The proof of Corollary~\ref{cor:cycle-free-bip} appears in Appendix~\ref{a:parti2}.
The following corollary is proven similarly.
\else
\coroapp

The proof of the following corollary is similar.
\fi
\begin{corollary}\label{cor:spanner}
There is a deterministic algorithm and a randomized algorithm in the \congest\ model that, given
$\eps \in (0,1)$, construct an $O(\poly(1/\eps))$-spanner of any unweighted, minor-free graph.
The spanner has $(1+O(\eps))n$ edges with probability $1$ in the deterministic algorithm, and with probability $1-\delta$ in the randomized algorithm.
The round complexity of the deterministic algorithm is $O(\poly(1/\eps)\log n)$ and of the randomized algorithm is $O(\poly(1/\eps)(\log(1/\delta)+\log^* n))$.
\end{corollary}

%% file: app.tex
\section{Proof of Lemma~\ref{cor:cycle}}


In order to 
prove Lemma~\ref{cor:cycle}, we introduce several notions.
These notions are very similar to those defined in~\cite[Chap.~7]{Even}, except that that there it was assumed that the graphs in question have vertex connectivity at least two. Here we do not make this assumption.

Let $H$ be a connected graph and let $C$ be a simple cycle in $H$. Consider the connected components of the graph resulting from removing the nodes in $C$. For each such connected component $D$, let $A(D)$ denote the subset of nodes on $C$ that neighbor nodes in $D$.
We refer to $A(D)$ as the {\em attachment\/} nodes of $D$ on $C$. Let $B(D)$ denote the subgraph induced by nodes of $D$ and $A(D)$, not including edges of $C$.
If $|A(D)| \geq 2$, then we refer to $B(D)$ as a {\em bridge\/}, and if $|A(D)| = 1$, then it is a {\em half-bridge\/}. We also refer to edges between pairs of nodes on $C$ as bridges. Two bridges $B$ and $B'$ are said to {\em interlace\/} if one of the following holds:
\begin{enumerate}
\item There are two attachments of $B$, $x$ and $y$, and two attachments of $B'$, $w$ and $z$, such that all four are distinct, and appear on $C$ in the order (say, clockwise) $x,w,y,z$.
\item There are three attachments common to $B$ and $B'$. That is, $|A(B) \cap A(B')| \geq 3$.
\end{enumerate}

The next lemma is a slight modification of Lemma~7.2 in~\cite{Even} (where here the graph $H$ in question is not assumed to have vertex connectivity at least two).
\begin{lemma}\label{lem:bridges}
Let $B_1,\dots,B_s$ be the set of bridges and half-bridges of a graph $H$ with respect to a simple cycle $C$.
Suppose that $C+B_i$ is planar for every $1 \leq i \leq s$ and that
no two bridges in the set interlace.
Then $C+B_1+\dots+B_s$ can be embedded in the plane so that all the bridges and half-bridges are inside $C$.
\end{lemma}

\begin{proof}
We prove the claim by induction on the number of nodes in $C+B_1+\dots+B_s$.
The base case is three nodes (there is just a cycle $C$ and no bridges).
For the induction step,
as in the proof of Lemma~7.2 in~\cite{Even}, since no two bridges interlace, there must be at least one bridge, $B_i$ for
which the following holds. If we consider the attachments of $B_i$ in clockwise order, $a_1,\dots,a_k$ (for $k\geq 2$), then
there is no other bridge $B_j$ with attachments (strictly) after $a_1$ and before $a_k$ on $C$.
\end{proof}

\medskip
As a corollary of Lemma~\ref{lem:bridges} we obtain:
\begin{corollary}\label{cor:C+B}
Let $B_1,\dots,B_s$ be the set of bridges and half-bridges of a graph $H$ with respect to a simple cycle $C$.
Suppose that $C+B_i$ is planar for every $1 \leq i \leq s$ and that the set of bridges can be partitioned into two subsets, such that within each subset no two bridges interlace.  Then $H$ is planar.
\end{corollary}
Building on Corollary~\ref{cor:C+B} we are now ready to prove Lemma~\ref{cor:cycle}.

\BPFOF{Lemma~\ref{cor:cycle}}
Consider the bridges of $H$ with respect to $C(u, v)$. They can be partitioned into two
pairwise noninterlacing subsets: one corresponding the bridges inside $C(u, v)$ and one corresponding to bridges outside $C(u, v)$ (recall that if the bridges interlace then we cannot embed them in one side of $C(u, v)$). Therefore $H$ is planar. It remains to show that there exists a planar embedding of $H$ which is consistent with $\tilde{\tau}$.
For each vertex $x$ in $C(u,v)$, we have the circular order of its edges in the planar embedding of $S_H^I(u,v)$ and $S_H^O(u,v)$, respectively.
We can 
merge this pair of orders into a single order that is consistent with $\tilde{\tau}$ simply by concatenating them (the only edges that these orders have in common are the edges on $C(u, v)$, which are the first and last edges in both orders).
Other vertices of $H$ are either in $S_H^I(u,v)$ or $S_H^O(u,v)$ (but not in both), so the ordering remains consistent.
\EPFOF
